\documentclass[STIX1COL]{WileyNJD-v2}

\usepackage{epsfig} 
\usepackage{amsmath,amssymb,amsfonts}
\usepackage{color}
\usepackage{balance}
\usepackage{nicefrac}
\usepackage{booktabs,tabularx}
\usepackage{multirow}
\usepackage{mathtools}
\usepackage{pdfrender}
\usepackage{trfsigns}

\usepackage{centernot}
\usepackage{color}
\usepackage[acronym]{glossaries}
\usepackage{glossaries-prefix}
\usepackage{graphicx}
\usepackage{tabu}
\usepackage[official]{eurosym}
\usepackage{subcaption}
\usepackage{ulem}

\graphicspath{{figures/}}

\newcommand{\norm}[1]{\left\lVert#1\right\rVert}

\newcommand{\defeq}{\coloneqq}
\newcommand{\rdefeq}{\eqqcolon}

\newcommand{\set}[2]{\left\{ #1\ \left| \ #2 \right. \right\}}

\newcommand{\R}{\mathbb{R}}

\newcommand{\B}{\mathbb{B}}

\newcommand{\N}{\mathbb{N}}

\newcommand{\bbS}{\mathbb{S}}
\newcommand{\mc}[1]{\mathcal{#1}}

\newcommand{\bs}{\boldsymbol}
\newcommand{\bsone}{\boldsymbol{1}}

\newcommand{\col}{\mathrm{col}}
\newcommand{\vertex}{\mathrm{vert}}
\newcommand{\rank}{\mathrm{rank}}

\newcommand{\normaltext}[1]{\textnormal{#1}}

\newcommand{\continuanceref}{}

\newtheorem{continuancex}{Example \continuanceref}
\newtheorem{example}{Example}
\newtheorem{standing}{Standing Assumption}

\newenvironment{continuance}[1]
{\renewcommand\continuanceref{\ref{#1}}\continuancex[Cont'd]}
{\endcontinuancex}

\usepackage{xcolor,calc}

\articletype{Research Article}%

\received{26 April 2016}
\revised{6 June 2016}
\accepted{6 June 2016}

\makeglossaries
\newacronym[
prefixfirst={a\ },
prefix={an\ }
]{LP}{LP}{linear program}
\newacronym[
prefixfirst={a\ },
prefix={an\ }
]{LTI}{LTI}{linear time-invariant}
\newacronym{CLF}{CLF}{control Lyapunov function}
\newacronym{MPC}{MPC}{model predictive control}
\newacronym{eMPC}{eMPC}{explicit model predictive control}
\newacronym[
prefixfirst={a\ },
prefix={an\ }
]{LQR}{LQR}{linear quadratic regulator}
\newacronym{mp-QP}{mp-QP}{multi-parametric quadratic program}
\newacronym{mp-LP}{mp-LP}{multi-parametric linear program}
\newacronym{QP}{QP}{quadratic program}
\newacronym[
prefixfirst={a\ },
prefix={an\ }
]{LICQ}{LICQ}{linear independence constraint qualification}
\newacronym[
prefixfirst={a\ },
prefix={an\ }
]{LMI}{LMI}{linear matrix inequality}
\newacronym{PWA}{PWA}{piecewise-affine}
\newacronym{PWA-NN}{PWA-NN}{piecewise-affine neural network}
\newacronym{SDP}{SDP}{semidefinite program}
\newacronym{MIP}{MIP}{mixed-integer program}
\newacronym{ReLU}{ReLU}{rectified linear unit}
\newacronym{PL-NN}{PL-NN}{piecewise linear neural network}
\newacronym{NN}{NN}{neural network}
\newacronym{iid}{i.i.d.}{independent and identically distributed}
\newacronym{wrt}{w.r.t.}{with respect to}
\newacronym[
prefixfirst={a\ },
prefix={an\ }
]{MI}{MI}{mixed-integer}
\newacronym[
prefixfirst={a\ },
prefix={an\ }
]{MILP}{MILP}{mixed-integer linear program}
\newacronym{LIQC}{LIQC}{linear independence constraint qualification}
\newacronym{KKT}{KKT}{Karush-Kuhn-Tucker}
\newacronym{ISS}{ISS}{input-to-state stable}

\begin{document}

\title{Robust stabilization of polytopic systems via fast and reliable neural network-based approximations}

\author[1]{Filippo Fabiani*}

\author[2]{Paul J. Goulart}

\authormark{FABIANI \textsc{et al}}

\address[1]{IMT School for Advanced Studies Lucca, Italy}

\address[2]{Department of Engineering Science, University of Oxford, United Kingdom}

\corres{*IMT School for Advanced Studies Lucca, Piazza S. Francesco 19, 55100, Lucca, Italy \email{filippo.fabiani@imtlucca.it}}

\abstract[Abstract]{We consider the design of fast and reliable \gls{NN}-based approximations of traditional stabilizing controllers for linear systems with polytopic uncertainty, including control laws with variable structure and those based on a (minimal) selection policy. Building upon recent approaches for the design of reliable control surrogates with guaranteed structural properties, we develop a systematic procedure to certify the closed-loop stability and performance of a linear uncertain system when a trained \gls{ReLU}-based approximation replaces such traditional controllers. 
First, we provide a sufficient condition, which involves the worst-case approximation error between \gls{ReLU}-based and traditional controller-based state-to-input mappings, {ensuring that the system is ultimately bounded within a set with adjustable size and convergence rate}. Then, we develop an offline, mixed-integer optimization-based method that allows us to compute that quantity exactly.
}


\maketitle
\normalem

\section{Introduction}\label{sec:intro}
Embodying a suitable compromise between model expressiveness and mathematical tractability, polytopic linear systems represent widely employed modelling paradigms able to capture structural uncertainties and parameter-varying dynamics \cite{MurJoh97}. Their analysis and control design, however, is frequently complicated by the presence of operational and physical constraints acting both on the control and state variables, thus possibly compromising the closed-loop performance of the system at hand.

\subsection{Traditional controllers for polytopic systems: advantages and disadvantages}
Among the available approaches, the concept of \emph{control invariant set} is one of the most exploited historically, since it ensures the existence of some feedback law able to steer the closed-loop trajectories of the uncertain system {within a prescribed state set} \cite{gutman1986admissible,blanchini1994ultimate,blanchini1995constrained,leyva2021stabilization}. This is traditionally achieved by associating a \gls{CLF} with the invariant set design, which for polytopic systems has been proven to be \emph{universal}, namely the stabilization of the linear uncertain system and the existence of a polyhedral \gls{CLF} can be used interchangeably \cite{blanchini1995nonquadratic}. With a specific focus on discrete-time polytopic systems, an admissible control policy that actually makes a polyhedral \gls{CLF} a suitable Lyapunov candidate for the closed-loop system is typically synthesized in two ways: through a \emph{variable structure} \cite{gutman1986admissible,nguyen2013implicit,nguyen2016explicit}, or a \emph{(minimal) selection-based} controller \cite{aubin2012differential}. We will also refer to these policies as traditional stabilizing controllers for linear uncertain systems.

Once fixed feasible control inputs at the vertices of the invariant set have been computed, a variable structure controller either takes a convex combination of those values by exploiting the vertex reconstruction of any state belonging to such a set, or coincides with a purely linear gain stemming from a \emph{triangulation}, i.e., a simplicial partition \cite{de2010triangulations}, of the underlying set. These methods therefore require one to solve a \gls{LP} online or to generate a lookup table to identify the region in which the current state resides. If the simplicial partition-based implementation is considered, then one has also to account for the complexity of the resulting invariant set, which is typically high \cite{blanchini1994ultimate,blanchini1995constrained,pluymers2005efficient,blanco2010efficient,athanasopoulos2010invariant,blanchini2015set}.  These methods can therefore require significant memory to store the vectors and/or matrices describing every simplicial partition and associated linear control gain. As a common drawback affecting both the implementations, however, fixing the input values at the vertices may result in poor control performance for the stabilization task.
A more sophisticated control method coincides with the selection-based policy. By requiring the online resolution of a nonlinear optimization problem, parametric in the current measured state, this method directly enforces a certain degree of contraction possessed by the \gls{CLF} at every control step.  While solving a numerical optimization problem online provides flexibility and performance guarantees, the real-time computational efforts required complicate its application in polytopic linear systems characterized by high sampling rates.

\subsection{Proposed approach and related work}
In this paper we aim to make the aforementioned control strategies more attractive from an implementation point of view by approximating their state-to-input behaviours through a \gls{NN} with \glspl{ReLU} \cite{hagan1997neural,Goodfellow-et-al-2016} resulting in a (continuous) \gls{PWA-NN} controller. The real-time evaluation of \gls{ReLU}-\glspl{NN} is known to be  computationally very inexpensive \cite{zhang2020near,schindler2020real}, thus allowing one to overcome the aforementioned practical limitations at the price of a possibly demanding offline training process.
Nevertheless, neural network-based methods have historically been used in uncertain system control design as a proxy for suboptimal control actions \cite{sznaier1992analog,hayakawa2008neural,liu2013neural}, robust feedback policies based on pole-placement \cite{wang1996multilayer,le2013robust}, or to solve frequently encountered linear matrix inequalities \cite{cheng2009simplified,guo2013zhang}. More recently, instead, an algorithm based on output range analysis methodologies was devised in \cite{karg2022guaranteed} to verify if the closed-loop operation of a linear uncertain system with a neural network controller guarantees safety for a set of initial conditions. The typically nonlinear and large-scale structure characterizing \glspl{NN}, however, generally complicate their analysis \cite{hagan1997neural,Goodfellow-et-al-2016}. Therefore, available approaches, despite commonly working well in practice, generally come without formal certificates of closed-loop stability and performance.

\subsection{Summary of contributions}
Capitalizing on the methodology proposed in \cite{fabiani2021reliably}, we take an optimization-based approach to develop analytical tools fulfilling such quest for theoretical guarantees of \gls{ReLU}-based approximations of traditional stabilizing controllers for polytopic systems. We develop a purely offline method based on the systematic construction and solution of a \gls{MILP} that allow us to certify the stability and performance of the closed-loop system when a trained \gls{ReLU}-based approximation replaces either the variable structure, or the selection-based controller.
The contributions made by the paper can hence be summarized as follows:

\begin{enumerate}
	\item We show how to guarantee the {(uniform, in a set) ultimate boundedness property \cite{blanchini1994ultimate}} of a discrete-time polytopic system when the \gls{ReLU} approximation replaces a traditional stabilizing controller. Specifically, by focusing on the approximation error between \gls{NN}-based and traditional controller-based state-to-input mappings, we establish a sufficient conditions involving the worst-case approximation error;
	\item While the variable structure controller amounts to a continuous \gls{PWA} mapping by construction, we characterize the geometric properties of the selection-based controller. Specifically, for the resulting nonlinear multi-parametric program, we show that:
	\begin{itemize}
		\item It produces a state-to-input mapping that also enjoys a continuos \gls{PWA} structure;
		\item We formulate a \gls{MILP} to compute both the output of the mapping, for a given state, and the associated Lipschitz constant exactly;
	\end{itemize}
	Note that these results are of standalone interest since, to the best of our knowledge, they are not directly available from the robust control of polytopic systems literature.
	\item We hence show that the approximation problem considered  is compatible with existing results from the machine learning literature on computing key quantities of a trained \gls{ReLU} network. We thus end up with a sufficient condition involving the optimal value of \pgls{MILP} to certify the reliability of the \gls{ReLU}-based controller;
	\item We discuss several practical aspects related with the training of the proposed \gls{ReLU}-based controllers, also providing bounds on the complexity they should have to represent a \gls{PWA} mapping exactly.
\end{enumerate}


While the high-level idea we adopt here qualitatively mirrors the one in \cite{fabiani2021reliably}, we stress that the results in this paper do not follow by direct application of those already available in the literature. Compared to \cite{fabiani2021reliably}, indeed, which was tailored for deterministic systems under the action of model predictive control policies, dealing with polytopic systems requires one to employ a different mathematical toolset and solution methodology, both to identify a sufficient condition to assess the training quality of a \gls{ReLU} network in mimicking traditional controllers (point 1 of the list above), and in characterizing the geometrical properties of these latter (especially the selection-based controller -- point 2). Putting together these two ingredients then allows us to recover a formulation that is compatible with the one available in the machine learning literature for assessing the training quality of a \gls{ReLU}-based controller (point 3). Only at this point we will finally make use of a technical result developed in \cite{fabiani2021reliably} (specifically, \cite[Prop.~5.2]{fabiani2021reliably}, showing that the norm of a matrix whose entries are affine in the state variable can be computed through \pgls{MILP}).

\subsection*{Notation}
$\N$, $\R$ and $\R_{\geq 0}$ denote the set of natural, real and nonnegative real numbers, respectively, $\N_0 \defeq \N \cup \{0\}$, $\N_\infty \defeq \N \cup \{+\infty\}$, and $\mathbb{B} \defeq \{0,1\}$. $\bbS^{n}$ is the space of $n \times n$ symmetric matrices and $\bbS_{\succ 0}^{n}$ ($\bbS_{\succcurlyeq 0}^{n}$) is the cone of positive (semi-)definite matrices. Bold $\bsone$ ($\bs{0}$) is a vector of ones (zeros). Given a matrix $A \in \R^{m \times n}$, $a_{i,j}$ is its $(i,j)$ entry, $A_{:,j}$ (resp., $A_{i,:}$) its $j$-th column ($i$-th row) and $A^\top$  its transpose. For $v \in \R^m$ and any $\mc{I} \subseteq \{1, \ldots, m\}$, $A_{\mc{I}}$ (resp., $v_{\mc{I}}$) denotes the submatrix (subvector) obtained by selecting the rows (elements) indicated in $\mc{I}$. For $A \in \bbS_{\succ 0}^{n}$, $\| v \|_A \defeq \sqrt{ v^\top A v }$. 	$\col(\cdot)$ stacks its arguments in column vectors or matrices of compatible dimensions. $A \otimes B$ is the Kronecker product of $A$ and $B$.{ We sometimes use $x(k+1)$, $k \in \N_0$, as opposed to $x^+$, to make time dependence explicit when describing the state evolution of discrete-time uncertain systems.}

We denote the norm over matrices in $\R^{m \times n}$ induced by an arbitrary norm $\|\cdot\|_\alpha$ over both $\R^n$ and $\R^m$ by $\|A\|_{\alpha} \defeq \textrm{sup}_{\|x\|_\alpha \leq 1} \ \|A x\|_\alpha = \textrm{sup}_{x \neq 0} \ \|A x\|_\alpha / \|x\|_\alpha$. For a set $\mc{S} \subseteq \R^n$, $|\mc{S}|$ is its cardinality, $\textrm{int}(\mc{S})$ its interior, $\textrm{cone}(\mc S)$ its conic hull and  $\vertex(\mc{S})$ its set of vertices (for polyhedral $\mc{S}$).
We call $\mc{S}$ a (polyhedral) C-set if it is compact, convex (polyhedral)  and $0 \in \textrm{int}(\mc{S})$, and we use the term polyhedral C-set or C-polytope interchangeably. Given a mapping $F : \R^n \to \R^m$, the local $\alpha$-Lipschitz constant over $\mc{S} \subseteq \R^n$ is denoted $\mc{L}_{\alpha}(F,\mc{S})$.  

\glsresetall

\section{On the approximation of traditional controllers for polytopic systems}\label{sec:problem_description}
We consider the control of a constrained, discrete-time linear system affected by polytopic uncertainty characterized by the following dynamics
\begin{equation}\label{eq:polytopic}
	x^+ = A(w) x + B(w) u.
\end{equation}
The state $x$ and input $u$ are constrained to sets $\mc{X} \subseteq \R^n$ and $\mc{U} \subseteq \R^m$, respectively.
The system dynamics are perturbed by a time varying exogenous signal $w \in \mc{W} \subseteq \R^M$, which determines at each time step the matrices $A(w) \defeq \sum_{i \in \mc{M}} w_i A_{i}$ and $B(w) \defeq \sum_{i \in \mc{M}} w_i B_{i}$, composed of generator matrices $\{(A_i,B_i)_{i \in \mc{M}}\}$ with $A_i \in \R^{n \times n}$ and $B_i \in \R^{n \times m}$, for all $i \in \mc{M} \coloneqq \{1, \ldots, M\}$.
We will assume throughout that the sets $\mc{X}$ and $\mc{U}$ are bounded, polyhedral and contain the origin, and that $\mathcal{W}$ is polyhedral{, specifically, an $M$-simplex,} with a known representation in vertex form $\mc W \defeq \set{w \in \R^M}{\bsone ^\top w = 1, w\ge 0}$.

The problem of designing a stabilizing controller for \eqref{eq:polytopic} has been well-studied, and a common approach (as in \cite{blanchini1994ultimate,blanchini1995constrained,pluymers2005efficient,blanco2010efficient,athanasopoulos2010invariant}) is to construct some control invariant set $\mc S \subseteq \mc X$ equipped with a feasible controller $\Phi : \R^n \to \mc{U}$ and \gls{CLF} $\Psi : \R^n \to \R$ \cite[Def.~2.29]{blanchini2015set} so that
\begin{equation}\label{eq:CLFdynamics}
	\Psi(A(w)x + B(w)\Phi(x)) \le \lambda \Psi(x), \text{ for all } w\in\mc W, x \in \mc S,
\end{equation}
where the value $\lambda \in (0,1)$ is the \emph{contraction factor} and determines the rate of convergence of the system to the origin.    We will assume throughout that a control-invariant, polyhedral C-set $\mc{S} \subseteq \mc{X}$ is available, with minimal representation
$$
\mc{S} \defeq \{x \in \R^n \mid F x \leq \bsone\}
$$
for $F \in \R^{p \times n}$, and that the set $\mc S$ and polyhedral function $\Psi(x) \defeq \|F x\|_\infty$ together satisfy \eqref{eq:CLFdynamics} in combination with some controller $\Phi(\cdot)$.   Specifically, the function $\Psi(\cdot)$ is the \emph{Minkowski} or \emph{gauge function} of $\mc S$ and can be written equivalently as
$$\Psi(x) = \underset{t \geq 0}{\textrm{min}} \ \{t \mid Fx \leq t \bsone\} = \underset{{j \in \mc{P}}}{\textrm{max}} \ \{F_{j,:} x\},
$$
with $\mc{P} \defeq \{1,\ldots,p\}$. We will denote the $\beta$-sublevel set of the gauge function $\Psi(\cdot)$ simply as $\beta \mc S = \set{x \in \R^n}{\Psi(x) \le \beta}$. Note that $\beta\le1$ to identify an actual sublevel set, since $\Psi(x) \in [0,1]$ for all $x\in\mc S $.

We will consider the problem of \emph{approximating} such a controller $\Phi(\cdot)$ using a \gls{ReLU}-based neural network, with the aim of reducing overall computational cost while retaining desirable stability and invariance properties.   We first describe briefly several common methods for designing the controller $\Phi(\cdot)$ given $\mc S$.

\subsection{Vertex-based control laws}\label{subsec:VertexController}
Methods for computing a polyhedral invariant C-set $\mc{S}$ typically associate with each vertex $x_v \in \vertex(\mc{S})$ some feasible action $u_v \in \mc{U}$ such that $A(w)x_v + B(w)u_v \in \lambda \mc S$.   Designing a controller over the whole of $\mc S$ is then a matter of exploiting the known control actions at the vertices.   Since $\mc S$ is polytopic it can be partitioned into a collection of simplices $\mc{T}^{(h)}$, each formed by $n$ vertices $\{x^{(h)}_{v}\}^n_{v = 1}$ and the origin (i.e.\ as a complete \emph{polyhedral fan} \cite[Def.~2.1.7]{de2010triangulations} -- see also Fig.~\ref{fig:vert_simplices} for an example).   If we define
$X^{(h)} \defeq [x^{(h)}_{1} \ x^{(h)}_{2} \ \cdots \ x^{(h)}_{n}]$ from the collection of non-zero vertices of the $h$-th simplex and $U^{(h)} \defeq [u_{1}^{(h)} \ u_{2}^{(h)} \ \cdots \ u_{n}^{(h)}]$ from the associated (feasible) control actions, then
\begin{equation}\label{eq:vertex_law_1}
	\Phi(x) \defeq U^{(h)}(X^{(h)})^{-1}x \ \text{ for all } x \in \mc T^{(h)},
\end{equation}
is stabilizing and satisfies \eqref{eq:CLFdynamics}. Determining the simplex in which any given state belongs can be done online, e.g., via a lookup table.

However, the complexity of the controller $\Phi(\cdot)$ in this case can be much greater than that of the generating set function $\mc S$, which is already complex in most cases \cite{blanchini2015set} -- see also Fig.~\ref{fig:th_PWA} in \S\ref{sec:geometric_char} for some numerical examples. In fact, for a state dimension larger than three choosing $n$ vertices on each facet of $\mc{S}$ to perform a triangulation with the origin leads to a total number of simplices that can be much larger than $p$ (i.e., the number of facets).	From \cite[\S 8.4]{de2010triangulations} it follows indeed that the \emph{size of the triangulation}, i.e., the number of simplices defining the partition, is bounded between $|\mc{V}_0|-n$ and $\mc{O}(|\mc{V}_0|^{\llcorner\nicefrac{n+1}{2}\lrcorner})$, where $\mc{V}_0 \defeq \vertex(\mc{S}) \cup \{0\}$.
This may require significant memory to store the matrices describing every simplicial partition $\mc{T}^{(h)}$ and associated control gain $\Gamma^{(h)}$, as well as requiring significant processing power for online evaluation.

\begin{figure}[t!]
	\centering
	\includegraphics[width=.6\columnwidth]{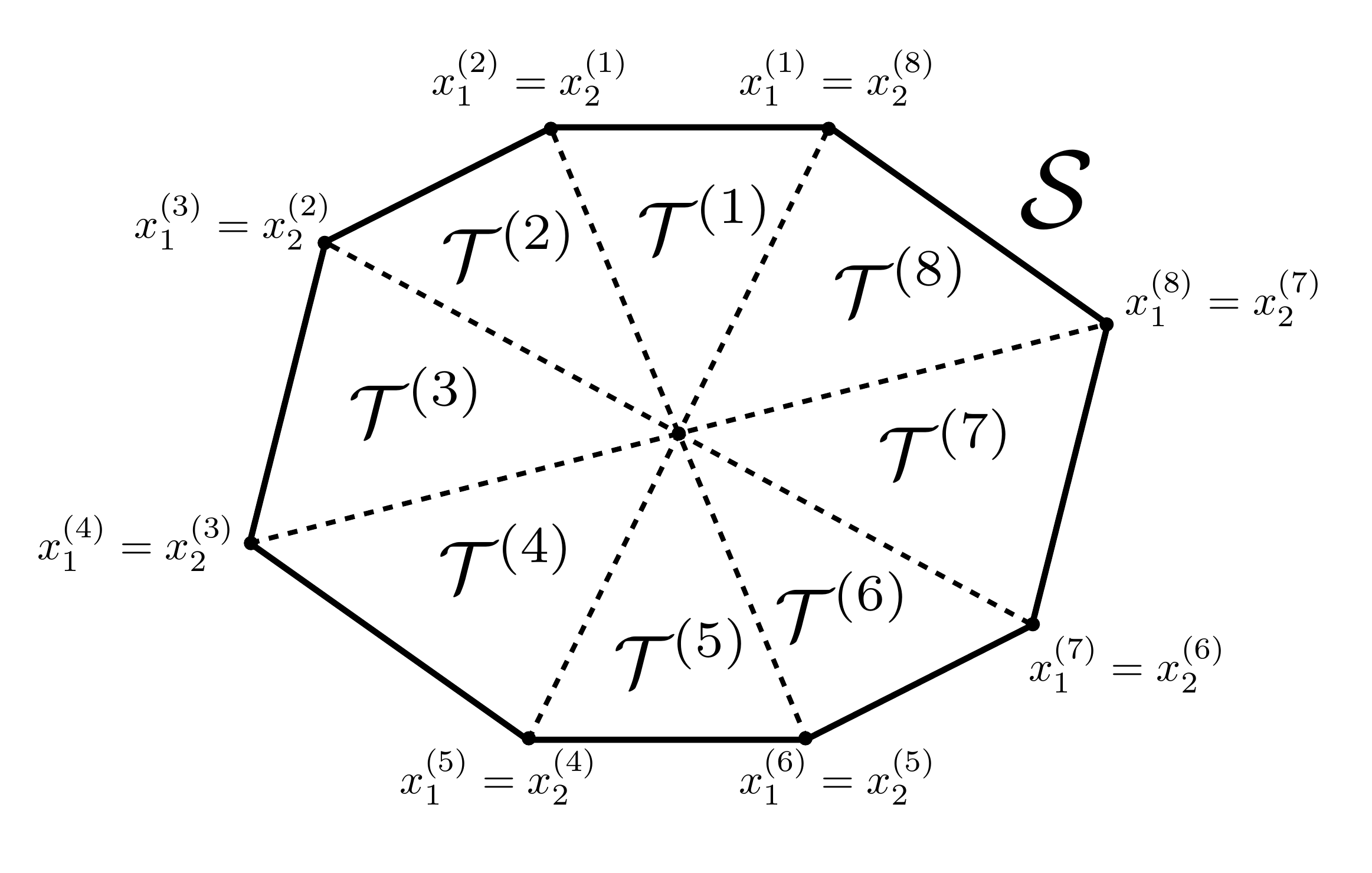}
	\caption{Two-dimensional triangulation of a symmetric C-polytope $\mc{S}$ into a sequence of eight simplices, $\mc{T}^{(h)}$, each one formed by two vertices $\{x^{(h)}_{v}\}^2_{v = 1}$, and the origin.}
	\label{fig:vert_simplices}
\end{figure}

A related approach is to instead define the controller $\Phi(\cdot)$ as
\begin{equation}\label{eq:vertex_law_2}
	\Phi(x) = \sum_{v =1}^{V} \gamma^\star_v(x) u_v,
\end{equation}
where $\gamma^\star(\cdot)$ solves
$$
\gamma^\star(x) = \underset{\gamma \in [0,1]^{V}}{\textrm{argmin}} \
\set{ \tfrac{1}{2}\|\gamma\|^2_2}{\sum\nolimits_{v = 1}^{V}  \gamma_v x_v = x,~\bsone^\top \gamma = 1},
$$
which amounts to a \gls{QP}, with $V \defeq |\vertex(\mc{S})|$.
Note that although the controllers \eqref{eq:vertex_law_1} and \eqref{eq:vertex_law_2} are both computed based on the vertex set $\vertex(\mc{S})$ and associated control inputs $u_v \in \mc{U}$, their values may differ since the optimal solution $\gamma^\star(x)$ does not necessarily return a convex combination of the vertices defining the simplex in which the current $x$ resides. In either case, however, the state-to-input mapping $x \mapsto \Phi(x)$ will be continuous and \gls{PWA}.

\subsection{Selection-based controller}\label{subsec:MinimalSelectionControl}
Optimization-based controllers more flexible than that of \eqref{eq:vertex_law_2} can be defined that account for
some type of control performance specifications.  To find a \emph{selection-based controller} \cite{aubin2012differential} directly enforcing the contraction condition \eqref{eq:CLFdynamics}, one can solve the following optimization problem online:
\begin{equation}\label{eq:online_control_1}
	\begin{aligned}
		\Phi(x) \defeq~ &\underset{v \in \mc{U}}{\textrm{argmin}} && \tfrac{1}{2} v^\top H v + x^\top P v\\
		& \textrm{ ~~s.t.} && \Psi(A_i x + B_i v) \leq  \lambda \Psi(x), \text{ for all } i \in \mc{M},
	\end{aligned}
\end{equation}
which turns into a \emph{minimal selection policy} in case $P=0$.
Under standard assumptions the problem \eqref{eq:online_control_1} will always be feasible  given the definition of a \gls{CLF}, and its optimal solution will be unique.  Unlike \eqref{eq:vertex_law_2}, however, it is not immediately obvious whether the controller $\Phi(\cdot)$ defined by \eqref{eq:online_control_1} will be \gls{PWA}.  We will prove in \S\ref{sec:geometric_char} that this remains the case.

As in the case of \eqref{eq:vertex_law_2}, the optimization-based controller $\Phi(\cdot)$ defined in \eqref{eq:online_control_1} is stabilizing but may require a prohibitive amount of computation for use in real-time applications with fast dynamics.   The practical limitations of characterizing a controller $\Phi(\cdot)$ designed by any of these methods lead us to the question of approximation -- is it possible to approximate these controllers with some proxy of substantially lower computational and storage cost while maintaining stability and a certain degree of performance?
To address this question, we focus on controllers implemented using \gls{ReLU} neural networks \cite{Goodfellow-et-al-2016}, which provide a natural means for approximating a continuous \gls{PWA} function $\Phi(\cdot)$ since the output mapping of such a network is itself \gls{PWA} continuous \cite{montufar2014number}.

Our approach parallels the development in \cite{fabiani2021reliably}, where we addressed the approximation of model predictive control policies for deterministic systems.   We ask whether the training of a \gls{ReLU}-based neural network to approximate a controller $\Phi(\cdot)$ has been sufficient to ensure that the network's output function $\Phi_{\textrm{NN}}(\cdot)$ will be stabilizing for \eqref{eq:polytopic}.  Our approach is based on the offline characterization of the error function $e(x) \defeq \Phi_{\textrm{NN}}(x) - \Phi(x)$ using \gls{MI} optimization, where $\Phi(\cdot)$ is a continuous \gls{PWA} law defined using any of \eqref{eq:vertex_law_1}, \eqref{eq:vertex_law_2} or \eqref{eq:online_control_1} (as we show in \S \ref{sec:geometric_char}). 
{We will obtain a condition on the optimal value of \pgls{MILP} sufficient to assure that the closed-loop system \eqref{eq:polytopic} under the action of $\Phi_\textrm{NN}(\cdot)$ is (uniformly) ultimately bounded within a set of adjustable size and (exponential) convergence rate, according to the following notion:

\begin{definition}\textup{(Uniform ultimate boundedness \cite[Def.~2.4]{blanchini1994ultimate})}\label{def:uub}
	Let $\mathcal{C}$ be a C-set, neighbourhood of the origin. The system in \eqref{eq:polytopic} with control $u(x)=\psi(x)$ is said to be ultimately bounded in $\mathcal{C}$, uniformly in $\mc S$, if for every initial condition $x(0)=x_0 \in \mc S$, there exists $K(x_0)$ such that, for all $k \geq K(x_0)$ and $w(k)\in \mc W$, $x(k) \in \mathcal{C}$.
	\hfill$\square$
\end{definition}
}

\section{Technical preliminaries}\label{sec:preliminaries}
We first recall some properties of continuous \GLS{PWA} functions and related properties of \gls{ReLU} networks:

\begin{definition}\textup{(Piecewise-affine mapping \cite[Def.~2.47]{rockafellar2009variational})}\label{def:pwa}
	$F:\mc{F} \to \R^m$ is a continuous \emph{piecewise-affine} mapping on the polyhedral domain $\mc{F} \subseteq \R^n$ if
	\begin{itemize}
		\item[(i)] $\mc{F}$ can be partitioned on a finite union of $R$ disjoint polyhedral sets, i.e. $\mc{F} \defeq \cup_{i =1}^R \mc{F}_i$ with $\normaltext{\textrm{int}}(\mc{F}_i) \cap \normaltext{\textrm{int}}(\mc{F}_j) = \emptyset$, for all $i \neq j$;
		\item[(ii)] $F(\cdot)$ is affine on each of the sets $\mc{F}_i$, i.e.,\
		$F(x) = F_i(x) \defeq G_i x + g_i$, $G_i \in \R^{m \times n}$, $g_i \in \R^m$, $\forall x \in \mc{F}_i$;
		{\item[(iii)]  It satisfies $G_i x + g_i = G_j x + g_j$ for all $x \in \mc F_i \cap \mc F_j$.}
	\end{itemize}
	Any mapping $F_i : \mc{F}_i \to \R^m$ is called a \emph{component} of $F(\cdot)$.
	\hfill$\square$
\end{definition}
From Definition~\ref{def:pwa} it is evident that any continuous \gls{PWA} mapping on $\mc{F}$ is Lipschitz continuous.  Recalling that the local $\alpha$\emph{-Lipschitz constant} of a mapping $F : \R^n \to \R^m$ on a set $\mc{F} \subseteq \R^n$ is
\[
\mc{L}_{\alpha}(F,\mc{F}) \defeq \underset{x \neq y \in \mc{F}}{\normaltext{\textrm{sup}}} \ \frac{\|F(x) - F(y)\|_{\alpha}}{\|x - y\|_{\alpha}},
\]
then it is easy to show that for any norm $\|\cdot\|_\alpha$ defined over both $\R^n$ and $\R^m$, a continuous piecewise affine mapping defined as in Definition~\ref{def:pwa} has Lipschitz constant
\begin{equation}\label{eq:Lipschitz_const_PWA}
	\mc{L}_{\alpha}(F,\mc{F}) = \underset{i \in \{1,\ldots,R\}}{\textrm{max}} \ \|G_i^\top\|_{\alpha}.
\end{equation}
\subsection{\gls{ReLU} networks are continuous \gls{PWA} mappings}\label{subsec:ReLU}
An $L$-layered, feedforward, fully-connected \gls{ReLU} network that defines a mapping $\Phi_{\textrm{NN}}: \R^n \to \R^{m}$ can be described by the following recursive equations across layers \cite{hagan1997neural}:
\begin{equation}\label{eq:RELU_NN}
	\left\{
	\begin{aligned}
		& x^{0} = x,\\
		& x^{j +1} = \textrm{max}(W^j x^j + b^j, 0), \ j \in \{0, \ldots, L-1\},\\
		& \Phi_{\textrm{NN}}(x) = W^L x^L + b^L,
	\end{aligned}
	\right.
\end{equation}
where $x^{0} = x \in \R^{n_0}$, $n_0 = n$, is the input to the network, $W^{j} \in \R^{n_{j + 1} \times n_j}$ and $b^j \in \R^{n_{j+1}}$ are the weight matrix and bias vector of the $(j+1)$-th layer, respectively. Note that this latter are typically defined during some offline training phase.  The width of the $(j+1)$-th layer, i.e.\ the number of neurons in the layer,  is therefore $n_{j + 1}$, with the $k$-th neuron in a given layer implementing the scalar function (also called  \emph{activation function}) $x^j \to \textrm{max} (W^{j}_{k,:} x^j + b^j_k,0)$.  The total number of neurons is thus $N \defeq \sum_{j = 1}^{L} n_{j} + m$, since $n_{L+1} = m$.

The function $\Phi_{\textrm{NN}}(\cdot)$ is a continuous \gls{PWA} mapping \cite{montufar2014number}, but direct computation of its Lipschitz constant using \eqref{eq:Lipschitz_const_PWA} is generally not practicable since it would require the explicit description of its \gls{PWA} representation as in Definition~\ref{def:pwa}.    However, it has been shown that under very mild assumptions the Lipschitz constant can be computed exactly through the solution of an \gls{MILP} \cite[Th.~5]{jordan2020exactly}.

\subsection{{On the ultimate boundedness} of polytopic systems with neural network controllers}\label{subsec:stability}
We now investigate the stability of the linear uncertain system in \eqref{eq:polytopic} with \gls{PWA-NN} controller $u(x) = \Phi_{\textrm{NN}}(x)$, defined as in \eqref{eq:RELU_NN}, trained to approximate some stabilizing controller $\Phi(\cdot)$ designed according to any of the techniques sketched in \S\ref{sec:problem_description}.    Define the approximation error as
$ e(x) = \Phi(x) - \Phi_{\textrm{NN}}(x)$, so that our problem amounts to verifying that the perturbed system
\begin{equation}\label{eq:perturbed_dyn}
	x^+ = A(w) x + B(w) \Phi_{\textrm{NN}}(x) = A(w) x + B(w) \Phi(x) + d(x,w),
\end{equation}
is robustly stable for all $w \in \mc{W}$, with $d(x,w) \defeq B(w) e(x)$.

Let us now fix some $b \in (0,1)$, whose exact value we will specify in the sequel, and define some values relating to $e(\cdot)$ over the whole of $\mc S$ and on the $b$-sublevel set $b \mc S$ of $\Psi(\cdot)$.
Given some $\alpha$-norm, define $\bar e_\alpha \defeq \textrm{max}_{x \in \mc S} \ \norm{e(x)}_\alpha$, $\alpha \in \{1, \infty\}$, as the worst case approximation error over $\mc S$.   Recall that $\mc{L}_\alpha(e, b \mc S)$ denotes the Lipschitz constant of $e(\cdot)$ over the sublevel set $b \mc S$.    We will assume that the training phase of our neural network controller enforces $e(0) = 0$, so that the latter quantity satisfies
\[
\|e(x)\|_\alpha \leq \mc{L}_\alpha(e, b \mc S) \|x\|_\alpha, \text{ for all } x \in b \mc S.
\]
In \S \ref{sec:certificates}, we will show how these values can be computed exactly using \gls{MI} optimization.

We can now derive a computable upper bound on a key quantity that allows us to preserve the stability and performance of the closed-loop system in \eqref{eq:polytopic} with $u(x) = \Phi_{\textrm{NN}}(x)$. Specifically, Proposition~\ref{prop:NN_stability} establishes that the trajectories of the perturbed system in \eqref{eq:perturbed_dyn} converge exponentially fast to a neighbourhood of the origin when the maximal approximation error $\bar{e}_{\alpha}$ is sufficiently small, thereby proving uniform ultimate boundedness in the spirit of Definition~\ref{def:uub}:

{
\begin{proposition}\label{prop:NN_stability}
		Suppose that $(b,\rho)$ are chosen so that $ b \in (0,1)$ and $\rho \in (\lambda, 1)$.  There exists a computable parameter $\zeta > 0$ such that, if $\bar{e}_{\alpha}  < \zeta$, then the system in \eqref{eq:polytopic} with \normaltext{\gls{PWA-NN}} contr0oller $u(x) = \Phi_{\normaltext{\textrm{NN}}}(x)$ is ultimately bounded in $b \mc S$, uniformly in $\mc S$. In addition, the resulting closed-loop trajectories are constrained in $\mc{S} \subset \mc{X}$, and converge exponentially fast to $b \mc S$ with rate $\rho$.
		\hfill$\square$
\end{proposition}
}
\begin{proof}
	We start by noting that the perturbation $d(x,w)$ acting on the system in \eqref{eq:perturbed_dyn} can be bounded for all $x \in \mc S$ and $w \in \mc W$ within the set $\mc D \defeq \set{d \in \R^n}{\norm{d}_\infty \leq \bar{d}}$, where $\bar{d} \defeq s \, \textrm{max}_{i \in \mc{M}} \{\|B_i\|_\infty\} \ \bar{e}_\alpha$, and $s > 0$ is some scaling factor accounting for the choice of $\alpha \in \{1, \infty\}$. In view of the subadditivity and positive homogeneity properties of the Minkowski functional $\Psi(x) = \|Fx\|_\infty$ \cite[Prop.~3.12]{blanchini2015set}, for all $x \in \mc S$ and $w \in \mc W$ we also have that
	$$
	\begin{aligned}
		\Psi(A(w) x + B(w) \Phi(x) + d(x,w)) &\le \Psi(A(w) x + B(w) \Phi(x)) + \Psi(d(x,w)),\\
		&\le \lambda \Psi(x) + \norm{F}_\infty \norm{d(x,w)}_\infty,\\
		&\le \lambda \Psi(x) + \tau \bar{e}_\alpha,
	\end{aligned}
	$$
	with $\tau \defeq s \norm{F}_\infty \textrm{max}_{i \in \mc{M}} \ \{\|B_i\|_\infty\}$. However, due to the effect of the disturbance the gauge function $\Psi(\cdot)$ is not guaranteed to decrease along the trajectories of \eqref{eq:perturbed_dyn}.
	On the other hand, since $d(x,w)$ is bounded we can fix some $b \in (0,1)$ to focus on a sublevel set $b \mc S \subset \mc S$ to show that it is robust positively invariant for the perturbed dynamics in \eqref{eq:perturbed_dyn}. In particular, for a given $x \in b \mc S$ we have that $\Psi(A(w) x + B(w) \Phi(x) + d(x,w)) < \Psi(x) \le b$ for all possible disturbances $d(x,w) \in \mc D$. We thus obtain $\Psi(A(w) x + B(w) \Phi(x) + d(x,w)) \le \lambda \Psi(x) + \tau \bar{e}_\alpha \le \lambda b + \tau \bar{e}_\alpha$, which is strictly smaller than $b$ if $b > \tau \bar{e}_\alpha / (1-\lambda)$. The perturbed one-step evolution then satisfies $x^+ = A(w) x + B(w) \Phi(x) + d(x,w) \in b \mc S$, for all $x \in b \mc S$ and $\infty$-norm bounded disturbances $\norm{d(x,w)}_\infty \leq \bar{d}$.
	
	Now suppose that $\bar{e}_\alpha \le (\rho - \lambda)b/\tau$ for some $\rho \in (\lambda, 1)$ so that $b \ge \tau \bar{e}_\alpha/(\rho - \lambda) > \tau \bar{e}_\alpha/(1 - \lambda)$ holds true. Then for any point $x \in \mc S \setminus b \mc S$ we have $\Psi(x) \leq b$, hence
	$$
	\begin{aligned}
		\Psi(A(w) x + B(w) \Phi(x) + d(x,w)) &\le \lambda \Psi(x) + \tau \bar{e}_\alpha,\\
		&\le \lambda \Psi(x) + (\rho - \lambda) \Psi(x) \leq \rho \Psi(x),
	\end{aligned}
	$$
	meaning that the perturbed dynamics will converge exponentially fast to $b \mc S$ with rate $\rho$, i.e., $\mc S \setminus b \mc S$ is a contractive set for \eqref{eq:perturbed_dyn}, and therefore robustly positively invariant. On the other hand, we shall also ensure that $b < 1$, thus finally obtaining
	\begin{equation}\label{eq:WC_bound}
		\bar{e}_{\alpha} < \frac{(\rho - \lambda)}{\tau} \rdefeq \zeta.
	\end{equation}
	{Thus, in case \eqref{eq:WC_bound} is satisfied, the conditions reported in Definition~\ref{def:uub} are met. In fact, since $\mc S \setminus b \mc S$ is $\rho$-contractive for the perturbed dynamics in \eqref{eq:perturbed_dyn}, some time instant $K = K(x(0)) \in \N_0$, $x(0)\in\mc S$, is guaranteed to exist finite so that \eqref{eq:perturbed_dyn} enters, and remains confined, in $b \mc S$. Hence, the system in \eqref{eq:polytopic} with \normaltext{\gls{PWA-NN}} controller $u(x) = \Phi_{\normaltext{\textrm{NN}}}(x)$ is ultimately bounded in $b \mc S$, uniformly in $\mc S$, where $b \mc S$ represents a neighbourhood of the origin with adjustable size.
	}
\end{proof}

In \S \ref{sec:certificates} we will then show how to compute the worst-case approximation error of $e(\cdot)$ exactly, thus providing a condition sufficient to certify the stability and performance of a \gls{ReLU}-based approximation of $\Phi(\cdot)$ as defined by any of \eqref{eq:vertex_law_1}--\eqref{eq:online_control_1}. A discussion on the complexity the \gls{NN} should have to make $\bar{e}_{\alpha}$ small enough will hence follow in \S \ref{subsec:complexity}.

\section{Geometric characterization of controllers for polytopic systems}\label{sec:geometric_char}
We now characterize a stabilizing control law $\Phi(\cdot)$ from a geometrical perspective. While both of the vertex-based policies $\Phi(\cdot)$ defined in \eqref{eq:vertex_law_1} or \eqref{eq:vertex_law_2} are known to produce a controller with \gls{PWA} structure, the structure underlying a selection-based controller $\Phi(\cdot)$ as defined in \S\ref{subsec:MinimalSelectionControl}  is less clear in view of the nonlinear constraints~in~\eqref{eq:online_control_1}.

In fact, the continuous \gls{PWA} structure of \eqref{eq:vertex_law_1} comes by construction, since it is defined directly over a simplicial partition,  while for  \eqref{eq:vertex_law_2} that structure can be proved by recognizing that the controller's definition amounts to that of a strictly convex \gls{mp-QP}, so that available results from \cite[Ch.~6.2, 6.3]{borrelli2017predictive} can be applied.   
In either case, the global error bound $\bar e_\alpha$ can be computed directly using the techniques in \cite{fabiani2021reliably}. The equivalent constants for the selection-based controller in \eqref{eq:online_control_1}, instead, are more difficult to determine.	

To this end, let the control input set be defined as $\mc U \defeq \set{v \in \R^m}{D v \le c}$, with $D \in \R^{\ell \times m}$, and $c \in \R^\ell$.   Computing the controller $\Phi(\cdot)$ in \eqref{eq:online_control_1} then amounts to solving the following optimization problem \cite[Eq.~(4.46)]{blanchini2015set}:
\begin{equation}\label{eq:online_control_2}
	\begin{aligned}
		\Phi(x)~=~&\underset{v}{\textrm{argmin}} && \tfrac{1}{2} v^\top H v + x^\top P v\\
		& \hspace{.1cm}\textrm{ ~s.t. } && C v \leq d + S(x),\\
	\end{aligned}
\end{equation}
where $H \in \mathbb{S}^m_{\succ 0}$, $P \in \R^{n \times m}$, and matrices $C \defeq \col(D, \col((FB_i)_{i \in \mc{M}})$, $d \defeq \col(c, \bs{0})$ and
$$
S(x) \defeq \left[
\begin{array}{c}
	\bs{0}\\
	\lambda \ \underset{j \in \mc{P}}{\textrm{max}} \ \{F_{j,:} x\} \otimes \bsone -  \col((F A_i)_{i \in \mc{M}}) x
\end{array}
\right].
$$
Note that standard methods for proving that the parametric solution of \eqref{eq:online_control_2} in $x$ is \gls{PWA} continuous can not be applied because the right hand side $S(x)$ of the inequalities is not affine.    We can be sure, however, that the problem has a solution for any $x \in \mc{S}$ since the associated Minkowski function $\Psi(\cdot)$ is a \gls{CLF} for the polytopic system in \eqref{eq:polytopic}. 
Next, we establish that $\Phi(\cdot)$ enjoys a continuous \gls{PWA} structure:

\begin{figure}[t!]
	\centering
	\includegraphics[width=.5\columnwidth]{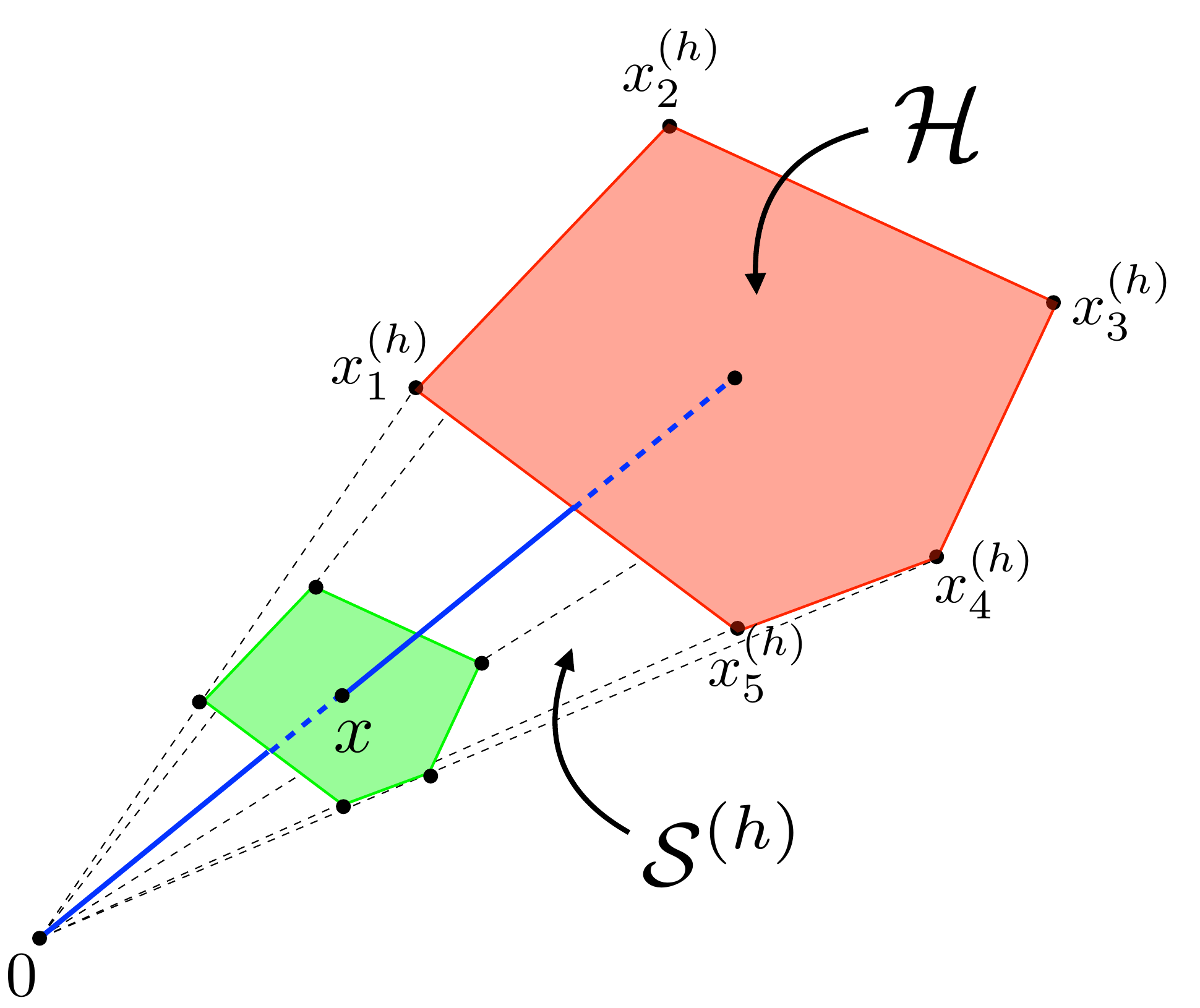}
	\caption{Three-dimensional schematic representation of the first part of the proof of Theorem~\ref{th:PWA}.
	}
	\label{fig:poly_cone}
\end{figure}

\begin{figure}[t!]
	\centering
	\includegraphics[width=.85\columnwidth]{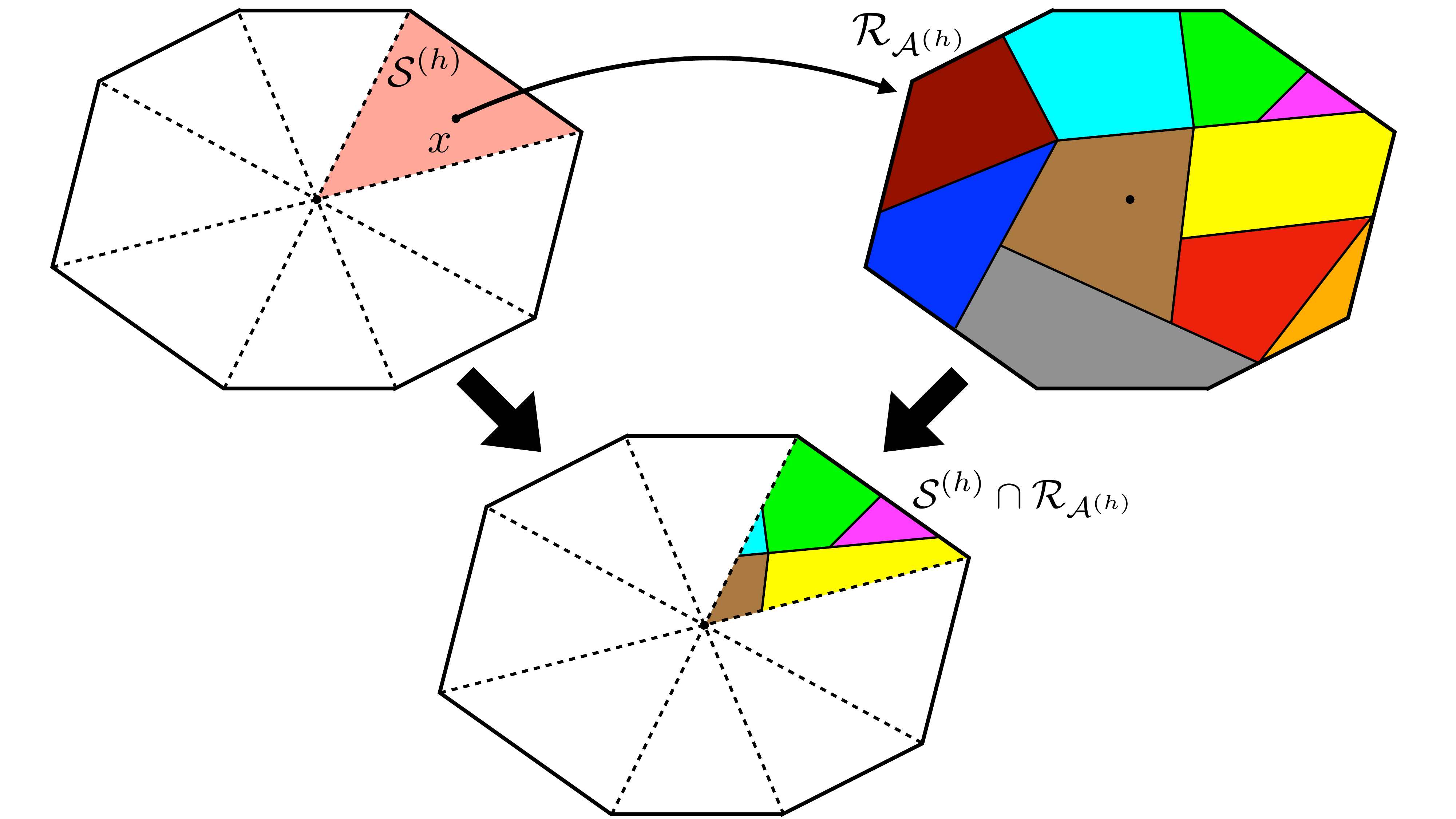}
	\caption{Pictorial representation of the proof of Theorem~\ref{th:PWA}. For any $x \in \mc{S}^{(h)}$ (shaded red area, top left figure) \eqref{eq:online_control_2} reduces to an \gls{mp-QP}, which provides an affine optimal solution in the given $x$ over a polyhedral partition of the whole of $\mc{S}$ (coloured regions, top right figure). Overall, this leads to a polyhedral partition of each sector, $\mc{S}^{(h)} \cap \mc{R}_{\mc{A}^{(h)}}$ (coloured regions inside the polyhedral sector, bottom figure), where $\mc{R}_{\mc{A}^{(h)}}$ denotes the \emph{critical region} of states $x$ associated with the set of active constraints $\mc{A}^{(h)} \subseteq \{1,\ldots,pM+\ell\}$, $\mc{R}_{\mc{A}^{(h)}} \coloneqq \set{x \in \mc S}{\mc{A}^{(h)}(x) = \mc{A}^{(h)}}$ and $\mc{A}^{(h)}(x)$ as defined in \eqref{eq:active_constraints}.}
	\label{fig:PWA}
\end{figure}

\begin{theorem}\label{th:PWA}
	The state-to-input mapping $x \mapsto \Phi(x)$ defined in \eqref{eq:online_control_2} is a \normaltext{\gls{PWA}} continuous mapping on $\mc{S}$.
	\hfill$\square$
\end{theorem}

Before proving Theorem~\ref{th:PWA} some preliminary considerations are needed. Suppose that there exists some region $\tilde{\mc{S}} \subseteq \mc{S}$ for which $S(x)$ is linear in $x$, i.e., $$S(x) = \tilde S x \ \text{ for all } x\in \tilde{\mc S},$$ for some matrix $\tilde S$. Then, for all $x \in \tilde{\mc{S}}$ the problem in \eqref{eq:online_control_2} becomes a \gls{QP} that can be written as
\begin{equation}\label{eq:online_control_compact}
	\left\{
	\begin{aligned}
		&\underset{v}{\textrm{min}} && \tfrac{1}{2} v^\top H v + x^\top P v\\
		& \hspace{0cm}\textrm{ s.t. } && C v \leq d +  \tilde S x,\\
	\end{aligned}
	\right.
\end{equation}
Without loss of generality we can additionally assume that $\rank(\tilde S) = n$, since otherwise \eqref{eq:online_control_compact} can be reduced to an equivalent form over a smaller set of parameters \cite{borrelli2017predictive}.

The problem \eqref{eq:online_control_compact} now looks like a standard \gls{mp-QP} whose solution can be computed parametrically in $x$.   Putting aside for the moment the issue of possible degeneracy, this parametric solution could even be computed over the whole of $\mc S$.   Our proof will therefore proceed as follows (see also Fig.~\ref{fig:PWA} for a schematic representation):
\begin{enumerate}
	\item[(i)] Partition $\mc S$ into a finite collection of disjoint polyhedral sets according to Definition~\ref{def:pwa}, where each region is associated with a different matrix $S^{(h)}$ such that $S(x) = S^{(h)}x$ for all $x \in \mc{S}^{(h)}$.
	\item[(ii)] For each $\mc{S}^{(h)}$, generate a continuous \gls{PWA} solution to an \gls{mp-QP} in the form \eqref{eq:online_control_compact} over the whole of $\mc S$.
	\item[(iii)] For each $\mc{S}^{(h)}$, intersect the set with the associated \gls{mp-QP} partition from (ii).  Assemble a solution over the whole of $\mc S$ from the resulting functions.
\end{enumerate}
Unlike the discussion in \S \ref{subsec:VertexController}, note that in this case each $\mc{S}^{(h)}$ does not necessarily stem from a triangulation procedure, i.e., it is not necessarily a simplex (unless $n = 2$, and hence one is obliged to choose $\mc S^{(h)} = \mc T^{(h)}$).
Before proceeding we must be careful to ensure that there is no issue with non-uniqueness of solutions to \eqref{eq:online_control_compact} over the generating sets $\mc{S}^{(h)}$ themselves.  To avoid pathological cases we will make a further assumption about \eqref{eq:online_control_compact} for any allowable choice of $\mc{S}^{(h)}$, the purpose of which will be evident in the proof of Theorem~\ref{th:PWA}. For some $x \in \mc{S}^{(h)}$ define the sets of \emph{active constraints} at a feasible point $v$ of \eqref{eq:online_control_compact} as:
\begin{equation}\label{eq:active_constraints}
	\mc{A}^{(h)}(x) \defeq \{i \in \{1, \ldots, pM+\ell \} \mid C_{i,:}  v - d_i - S^{(h)}_{i,:}  x = 0\}.
\end{equation}
\begin{standing}\label{standing:LICQ}\textup{(Linear independence constraint qualification \cite[Def.~2.1]{borrelli2017predictive})}\label{def:LICQ}
	For any $x \in \mc{S}^{(h)}$, the \normaltext{\gls{LICQ}} is said to hold at a feasible point $v$ of
	\eqref{eq:online_control_compact} if $C_{\mc{A}^{(h)}(x), :}$ has linearly independent rows.
	For all regions $\mc{S}^{(h)} \subseteq \mc{S}$ over which $S(x)$ is linear, the \normaltext{\gls{LICQ}} is assumed to hold for the resulting \normaltext{\gls{mp-QP}} as in \eqref{eq:online_control_compact} for all $x \in \mc{S}^{(h)}$.
	\hfill$\square$
\end{standing}
The \gls{LICQ} is sufficient to exclude the case where more than $m$ constraints are active at a given feasible point $v$, thereby avoiding primal degeneracy \cite[\S 4.1.1]{bemporad2002explicit}. For any region $\mc{S}^{(h)} \subseteq \mc{S}$ for which $S(x)$ is linear, under \gls{LICQ}  the problem dual to the resulting \gls{mp-QP} in \eqref{eq:online_control_compact} will be a strictly convex program whose solution is characterized by a unique choice of Lagrange multipliers.
We are now in a position to prove the main result:

\textit{Proof of Theorem~\ref{th:PWA}}:
Suppose that $\mc {H}$ is the facet of $\mc S$ generated by the $h$-th row of the inequality $Fx \le \bf 1$ that defines $\mc S$.   Each such facet generates a polyhedral cone that can be defined as $\textrm{cone}(\mc H) \defeq \set{\lambda x}{ x\in \mc H, \ \lambda\ge 0}$.    Define $\mc{S}^{(h)} \defeq \mc S \cap \textrm{cone}(\mc H)$, which is a polyhedral set with one vertex at the origin and the remaining vertices on $\mc H$ (see, e.g., Fig.~\ref{fig:poly_cone}).  It is then easily shown that
$
\textrm{argmax}_{j\in \mc{P}} \ \{F_{j,:} x \}
$
is the same for all $x\in \mc{S}^{(h)}$, and therefore there exists some matrix $S^{(h)}$ such that $S(x) = S^{(h)} x$ for all $x\in \mc{S}^{(h)}$. Specifically, $S^{(h)} = \col(\bs{0}, \lambda F_{h,:} \otimes \bs{1} - \col((F A_i)_{i \in \mc{M}}))$.

Fix now any such set $\mc{S}^{(h)}$ and consider the \gls{mp-QP}
\begin{equation}\label{eq:online_control_compact_2}
	\begin{aligned}
		v^{\star,(h)}(x) ~=~&\underset{v}{\textrm{argmin}} && \tfrac{1}{2} v^\top H v + x^\top P v\\
		& \hspace{0cm}\textrm{ ~~s.t. } && C v \leq d +  S^{(h)} x,\\
	\end{aligned}
\end{equation}
Since the problem \eqref{eq:online_control_2} is feasible for all $x\in \mc S$, the set of points $x$ for which \eqref{eq:online_control_compact_2} is feasible must be at least as large as $\mc S^{(h)} \subseteq \mc S$ since the constraints in \eqref{eq:online_control_compact_2} are identical to those in \eqref{eq:online_control_2} over the set $\mc S^{(h)}$.   Considering \eqref{eq:online_control_compact_2} as an \gls{mp-QP} in $x$ therefore produces a continuous \gls{PWA} solution over a polyhedral partition of a set that covers $\mc S^{(h)}$.  Standing Assumption~\ref{standing:LICQ} ensures that this solution is unique everywhere on $\mc S^{(h)} \subseteq \mc S$.

Finally, we can form the overall \gls{PWA} solution to \eqref{eq:online_control_2} by combining the solutions over the regions $\mc S^{(h)}$, which completely cover $\mc{S}$, i.e.
\[
\Phi(x) = v^{\star,(h)}(x) \ \text{ for all } x\in \mc{S}^{(h)}, \, h \in \mc P.
\]
This function is continuous \gls{PWA} because each $v^{\star,(h)}(x)$ is \gls{PWA} continuous on $\mc S^{(h)}$ and the sets $\mc S^{(h)}$ form a polyhedral partition of $\mc S$, according to Definition~\ref{def:pwa}.
\hfill$\blacksquare$

\begin{figure*}
	\centering
	\begin{subfigure}[b]{0.19\columnwidth}
		\centering
		\includegraphics[width=\columnwidth]{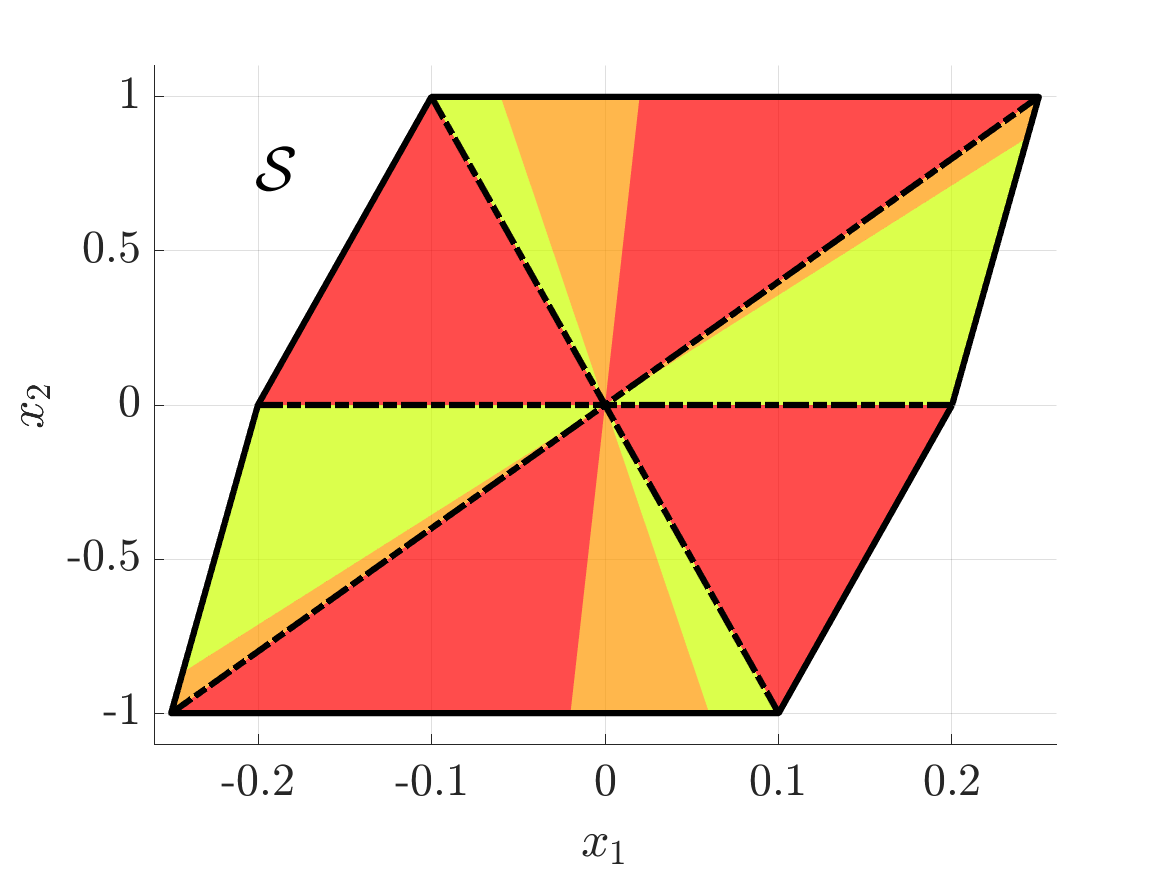}
		\caption{}
		\label{fig:a}
	\end{subfigure}
	\hfill
	\begin{subfigure}[b]{0.19\columnwidth}
		\centering
		\includegraphics[width=\columnwidth]{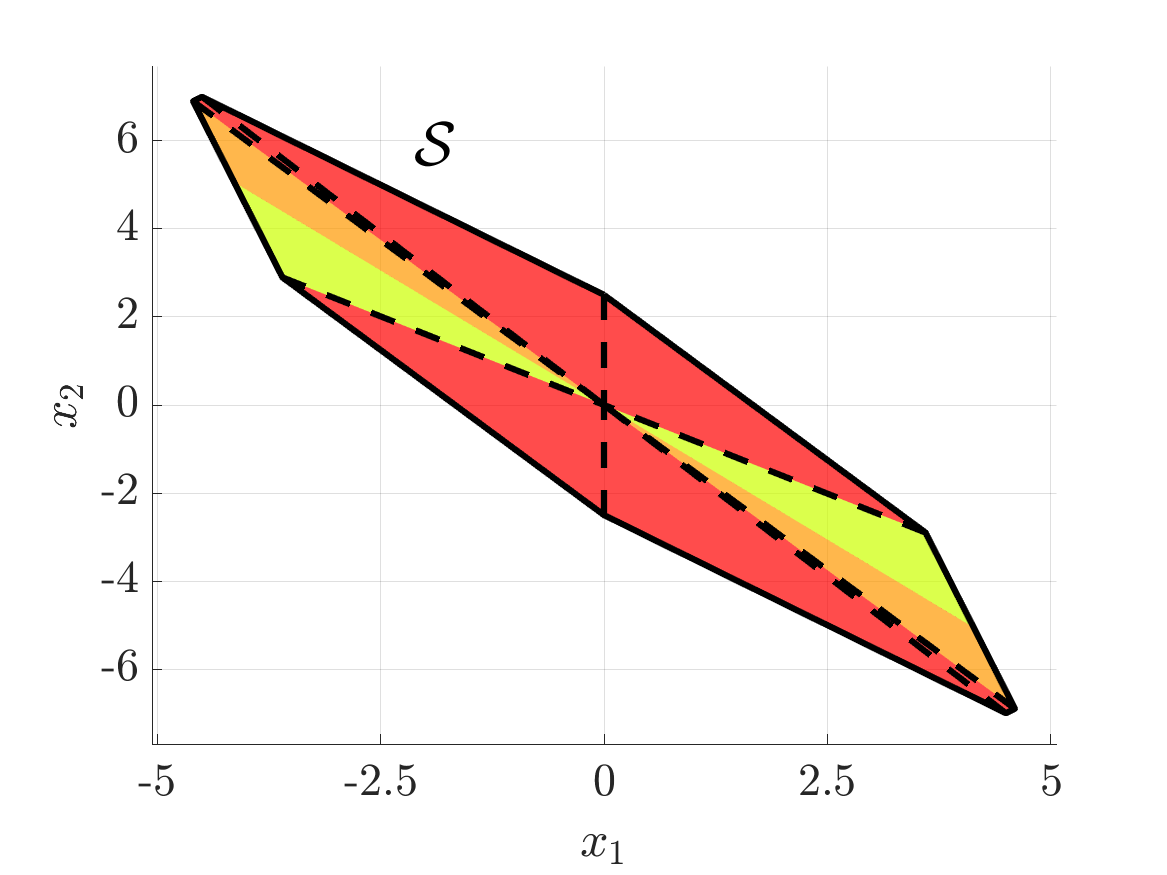}
		\caption{}
		\label{fig:b}
	\end{subfigure}
	\hfill
	\begin{subfigure}[b]{0.19\columnwidth}
		\centering
		\includegraphics[width=\columnwidth]{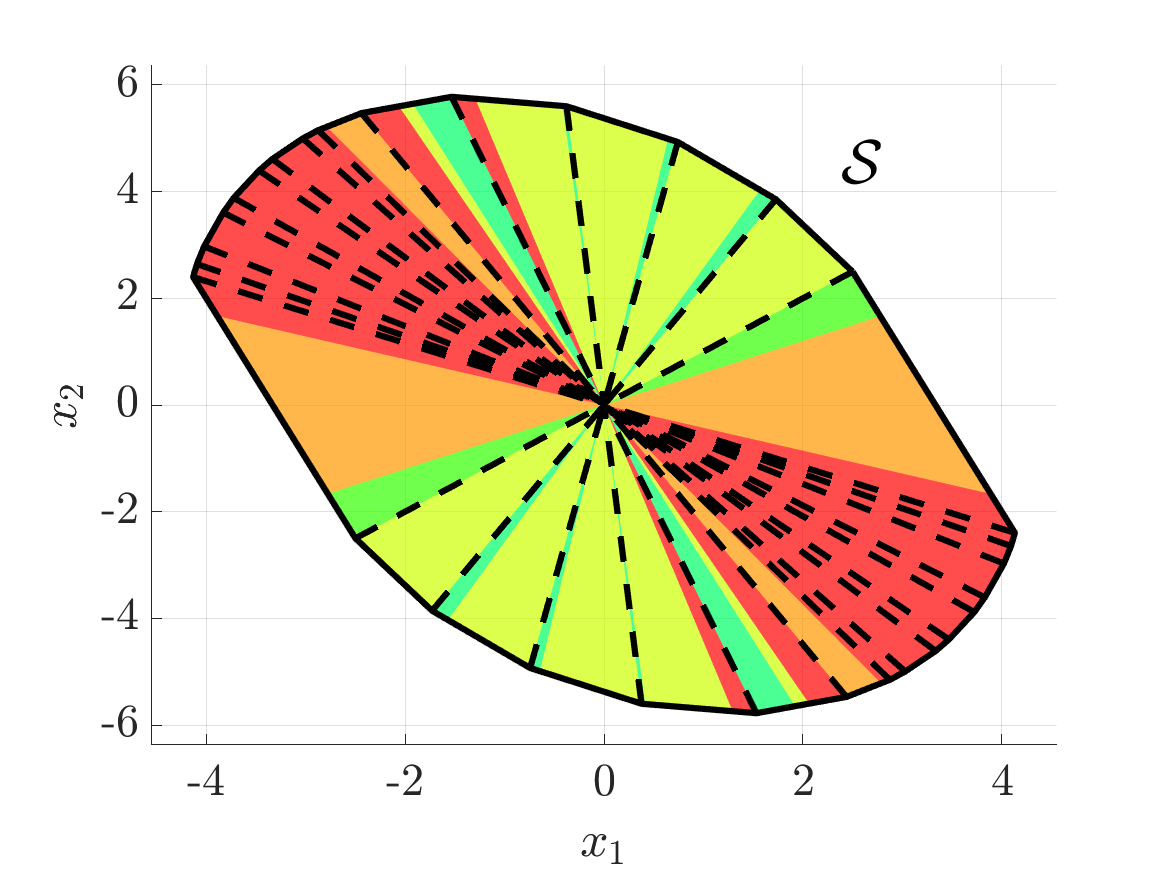}
		\caption{}
		\label{fig:c}
	\end{subfigure}
	\hfill
	\begin{subfigure}[b]{0.19\columnwidth}
		\centering
		\includegraphics[width=\columnwidth]{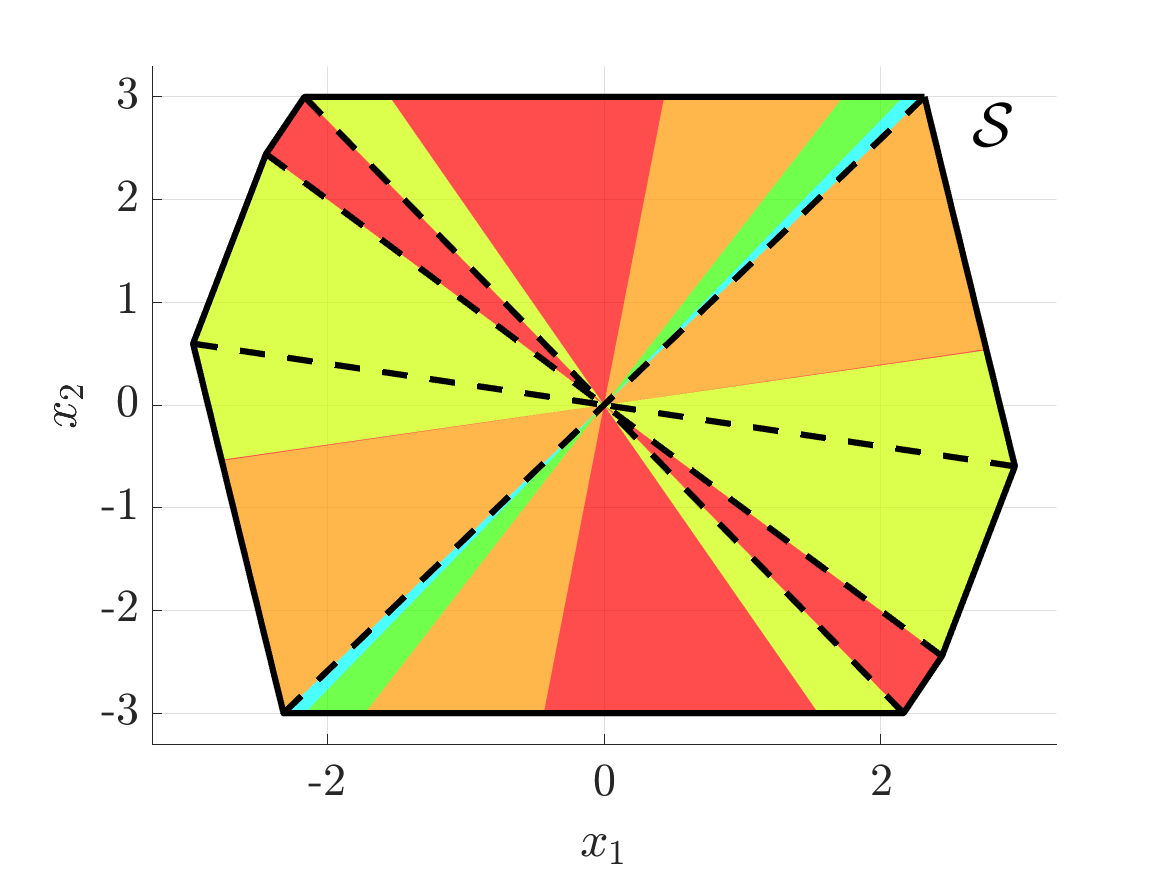}
		\caption{}
		\label{fig:d}
	\end{subfigure}
	\hfill
	\begin{subfigure}[b]{0.19\columnwidth}
		\centering
		\includegraphics[width=\columnwidth]{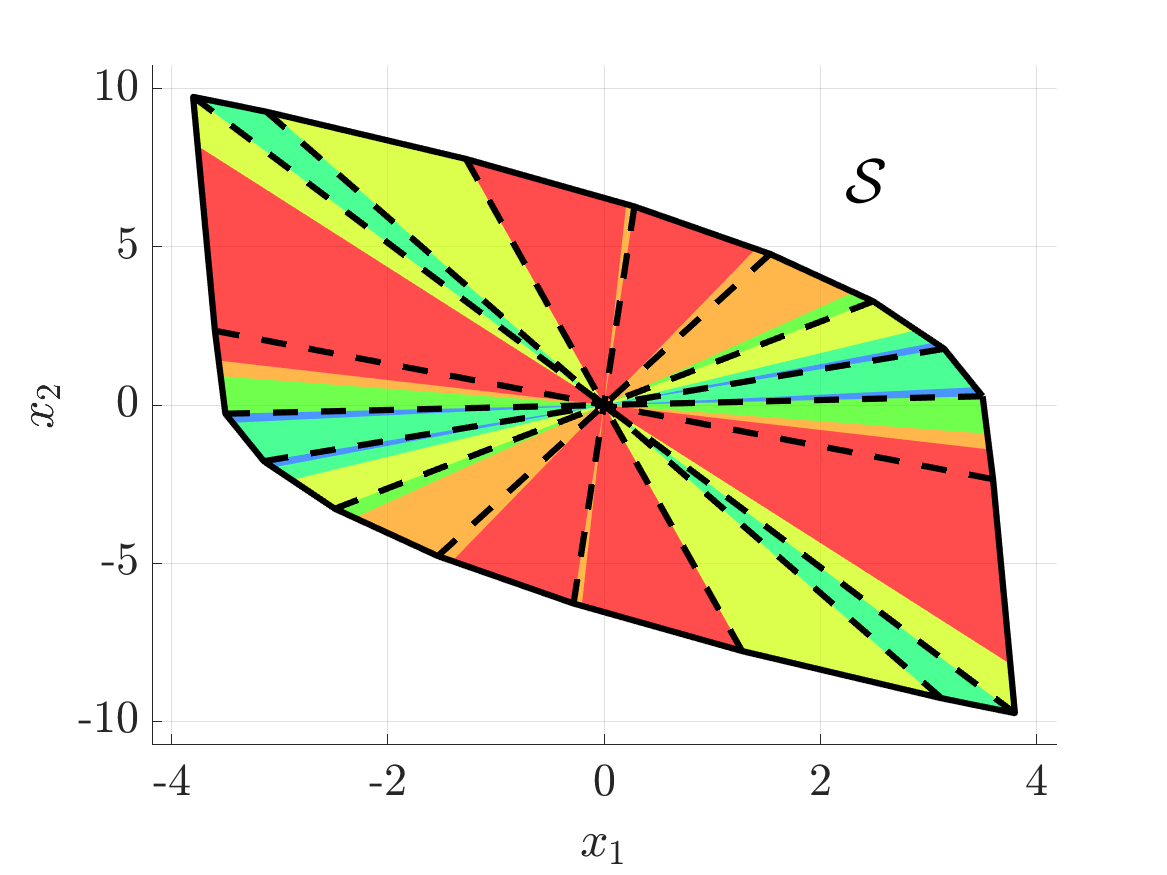}
		\caption{}
		\label{fig:e}
	\end{subfigure}
	\\
	\begin{subfigure}[b]{0.19\columnwidth}
		\centering
		\includegraphics[width=\columnwidth]{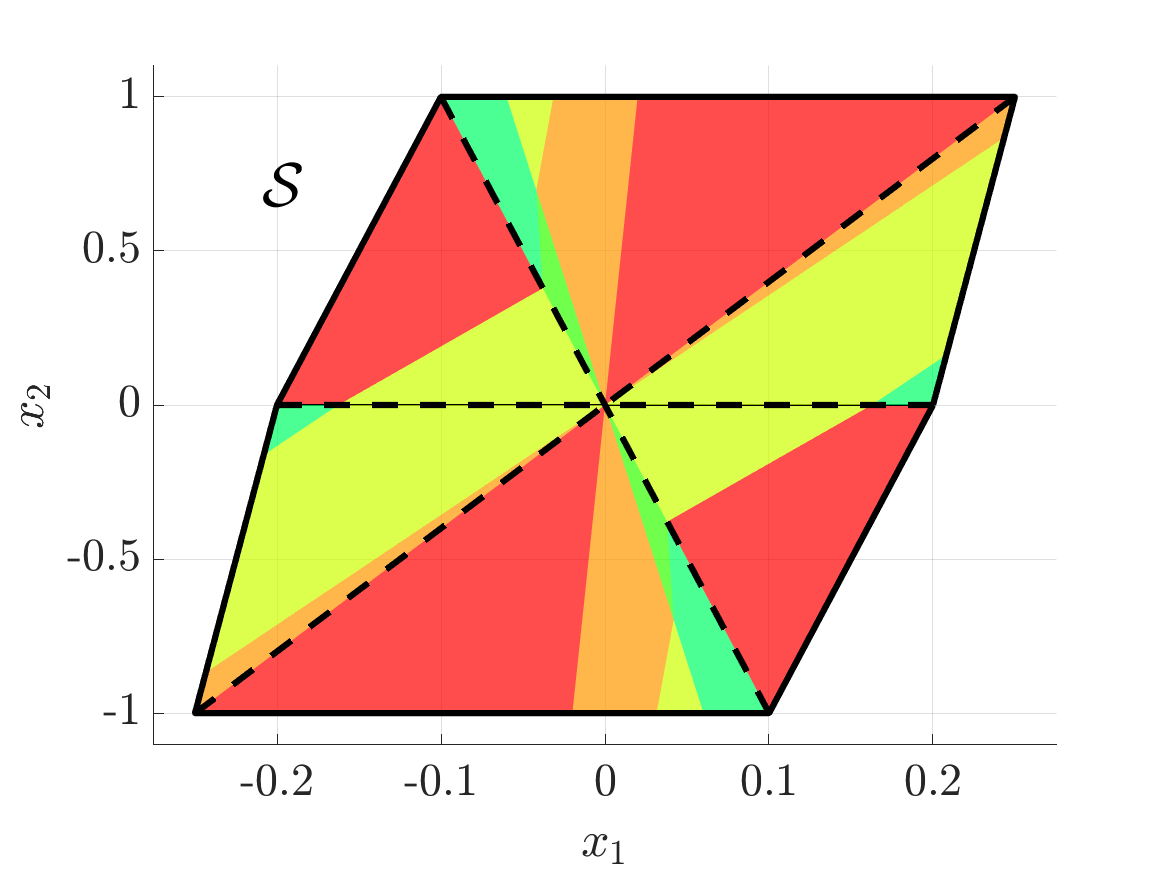}
		\caption{}
		\label{fig:f}
	\end{subfigure}
	\hfill
	\begin{subfigure}[b]{0.19\columnwidth}
		\centering
		\includegraphics[width=\columnwidth]{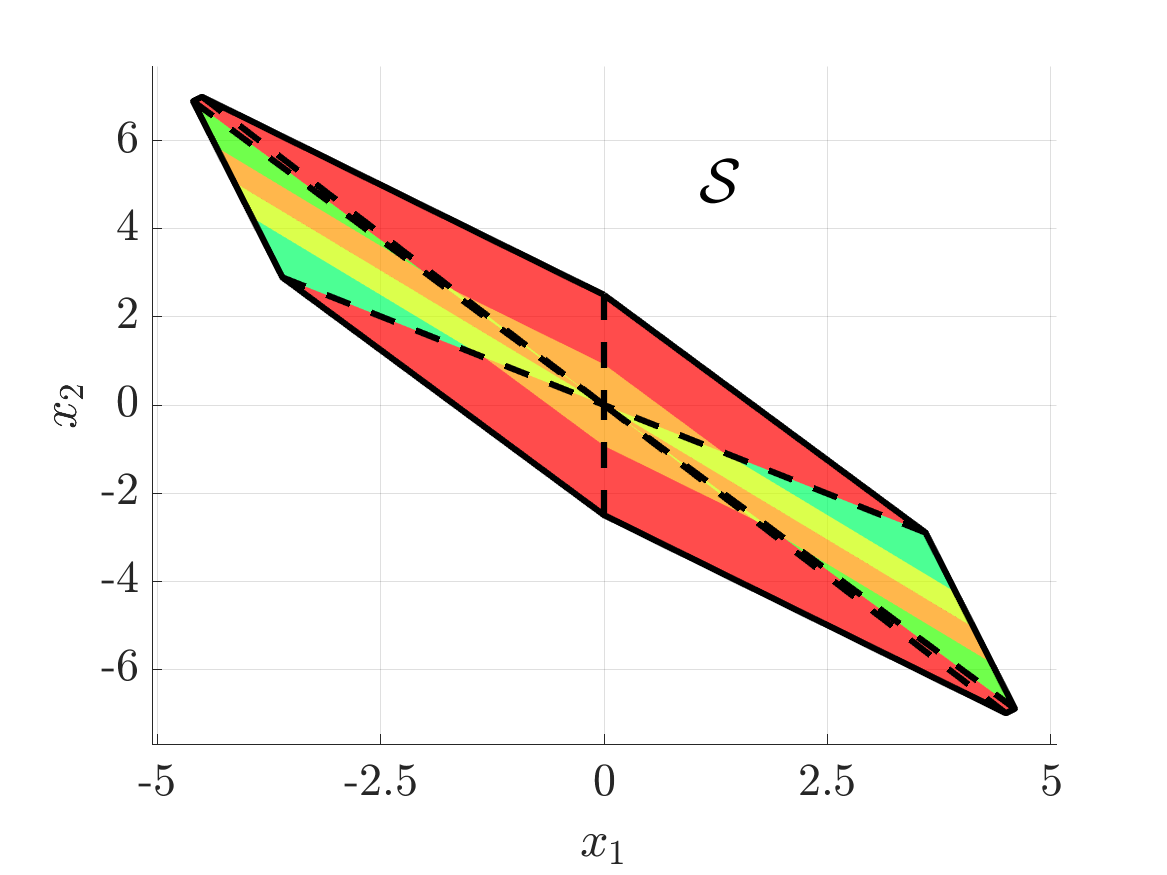}
		\caption{}
		\label{fig:g}
	\end{subfigure}
	\hfill
	\begin{subfigure}[b]{0.19\columnwidth}
		\centering
		\includegraphics[width=\columnwidth]{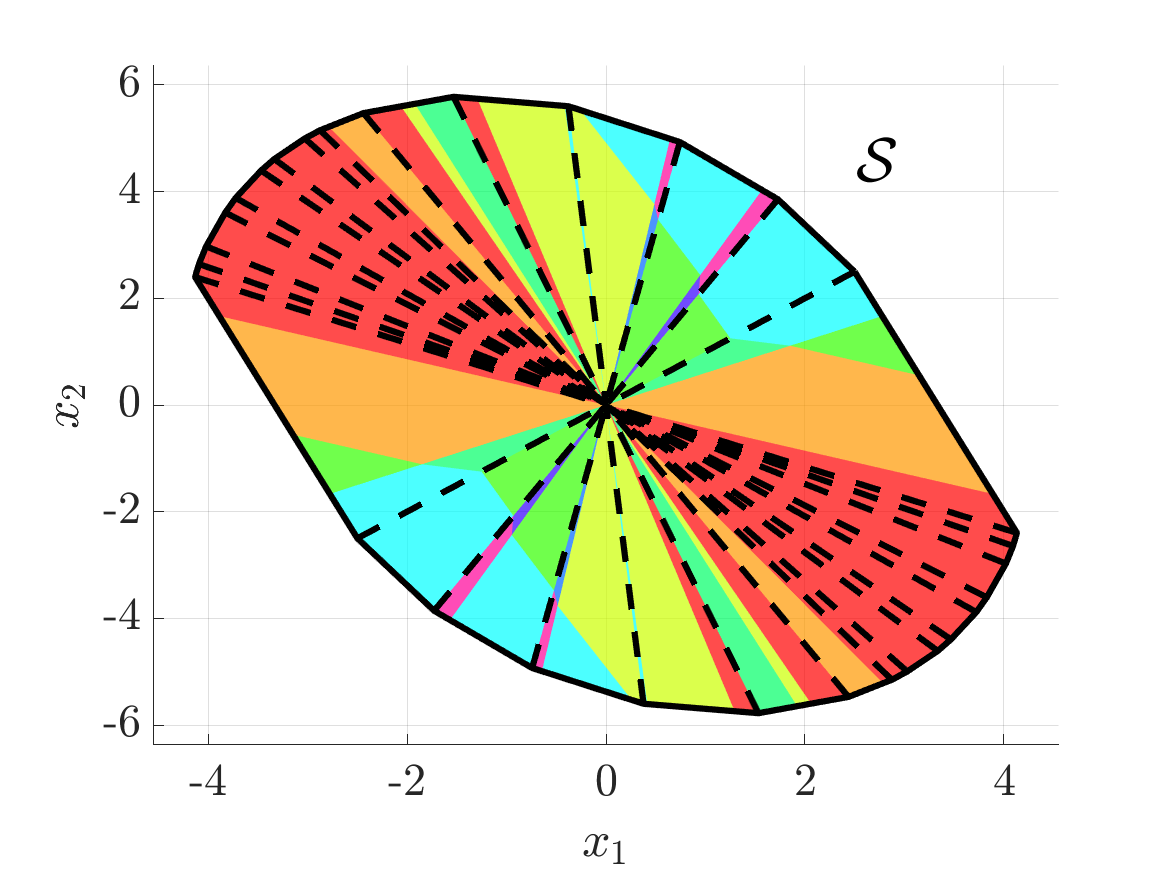}
		\caption{}
		\label{fig:h}
	\end{subfigure}
	\hfill
	\begin{subfigure}[b]{0.19\columnwidth}
		\centering
		\includegraphics[width=\columnwidth]{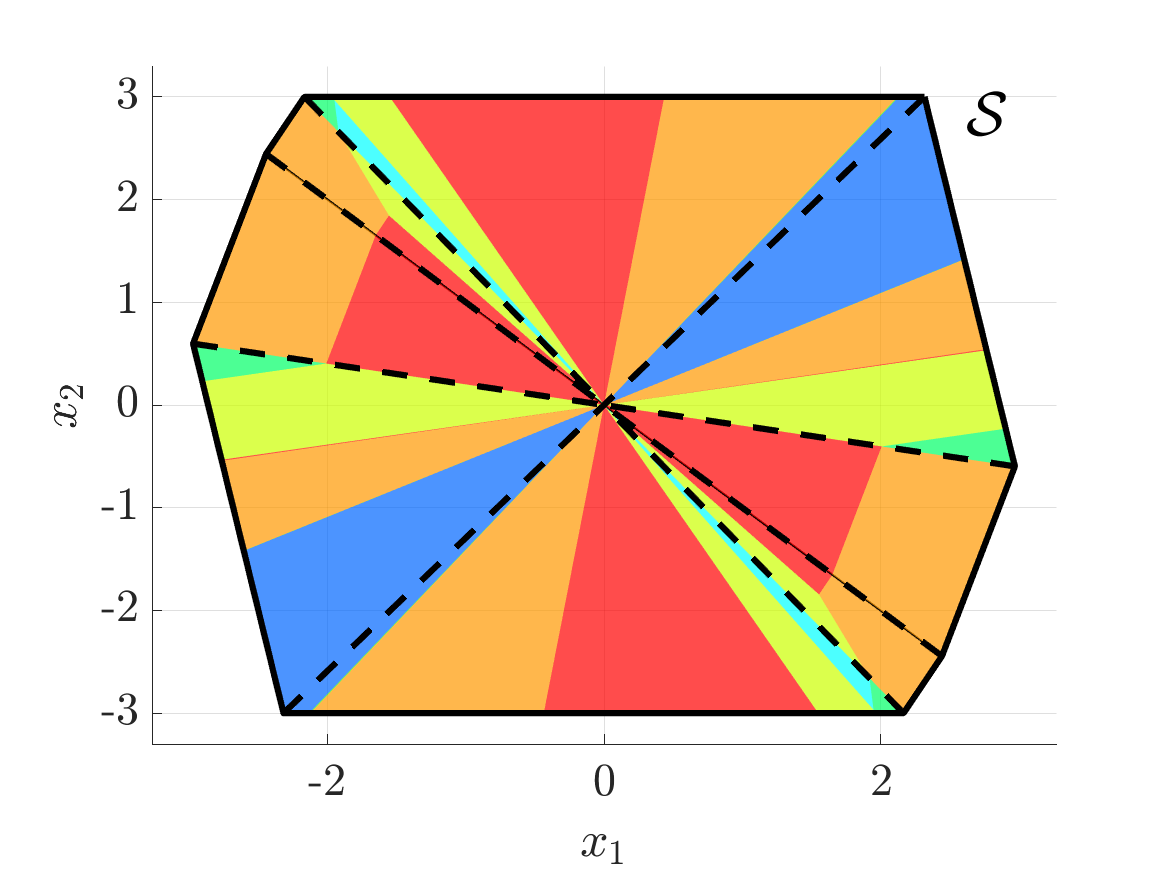}
		\caption{}
		\label{fig:i}
	\end{subfigure}
	\hfill
	\begin{subfigure}[b]{0.19\columnwidth}
		\centering
		\includegraphics[width=\columnwidth]{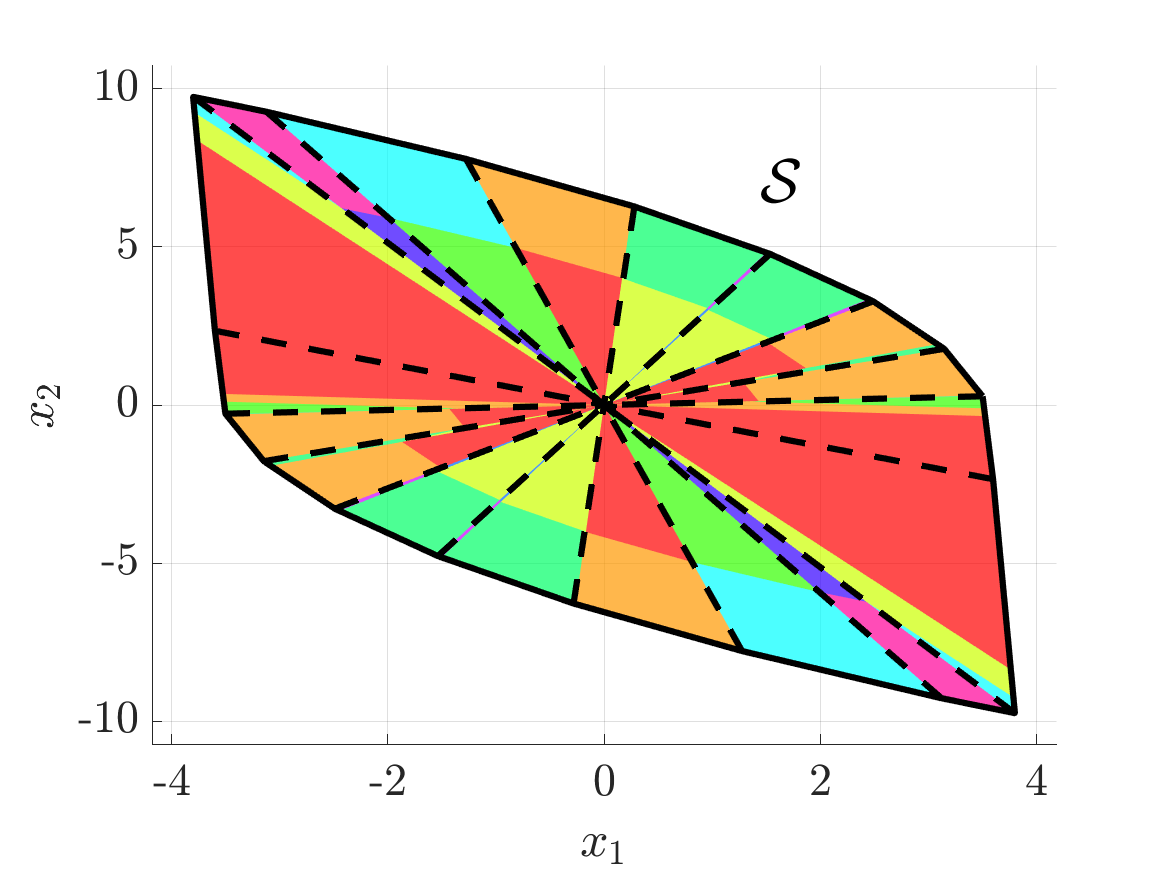}
		\caption{}
		\label{fig:j}
	\end{subfigure}
	\caption{Examples of partitions characterizing the \gls{PWA} mapping $\Phi(\cdot)$ as defined in \eqref{eq:online_control_1}. The dashed black lines denote different polyhedral sectors of $\mc{S}$, in turn partitioned into further polyhedral regions (coloured areas).}
	\label{fig:th_PWA}
\end{figure*}

\begin{table}[!t]
	\caption{Examples from the literature}\label{tab:examples}
	\centering
	\begin{tabular}{ccccccc}
		\toprule
		Reference & Ex. & $|\mc{P}|$ & $N^r$ & $\lambda$ & $\mc{U}$\\
		\midrule
		\multirow[c]{2}{*}{\cite[Ex.~7.41]{blanchini2015set}} & (a) & \multirow[c]{2}{*}{$6$}   & $14$ & \multirow[c]{2}{*}{$0.8$} & $\R$\\
		& (f)  & & $21$ & & $|u|\leq 2$\\
		\midrule
		\multirow[c]{2}{*}{\cite[Ex.~4.3]{blanchini1995constrained}} & (b) & \multirow[c]{2}{*}{$8$}   & $12$ & \multirow[c]{2}{*}{$0.92$} & $\R$\\
		& (g)  & & $26$ & & $|u|\leq 1.5$\\
		\midrule
		\multirow[c]{2}{*}{ \cite[\S V]{pluymers2005efficient}} & (c) & \multirow[c]{2}{*}{$30$}  & $48$ & \multirow[c]{2}{*}{$0.95$} & $\R$\\
		& (h)  & & $61$ & & $|u|\leq 1$\\
		\midrule
		\multirow[c]{2}{*}{\cite[Ex.~5.11]{blanchini2015set}} & (d) & \multirow[c]{2}{*}{$8$}  &  $24$ & \multirow[c]{2}{*}{$0.61$} & $\R$\\
		& (i)  & & $40$ & & $|u|\leq 0.5$\\
		\midrule
		\multirow[c]{2}{*}{\cite[Ex.~2.6]{nguyen2012constrained}} & (e) & \multirow[c]{2}{*}{$18$}  & $38$ & \multirow[c]{2}{*}{$0.96$} & $\R$\\
		& (j)  & & $66$ & & $|u|\leq 1$\\
		\bottomrule
	\end{tabular}
\end{table}

\begin{example}\label{ex:numerics}
	Fig.~\ref{fig:th_PWA} shows several numerical examples available in the literature with all models summarized in Tab.~\ref{tab:examples}. Specifically, they show the polyhedral partition of each sector in which a control invariant set for a polytopic system as in \eqref{eq:polytopic} can be divided, thus illustrating the result of Theorem~\ref{th:PWA} numerically. The partitions are obtained by solving \eqref{eq:online_control_2} with $P = 0$ and are computed via the Multi-Parametric Toolbox \cite{herceg2013multi}. For the numerical examples \normaltext{(a)--(e)} we set $\mc{U} = \R$, which results in purely proportional controllers as the matrix $D$ and vector $c$ do not appear in the constraints in \eqref{eq:online_control_2}. Finally, note that in both cases the number of polyhedral regions produced by \eqref{eq:online_control_2}, $N^r$, is always greater than the number of facets of the considered control invariant set $\mc{S}$, illustrating that  significant memory may be required to store all the involved matrices/vectors for every region.
	\hfill$\square$
\end{example}
The fact that $x \mapsto \Phi(x)$ as defined in \eqref{eq:online_control_1} (or \eqref{eq:online_control_2}) is Lipschitz continuous was already known from standard results in viability theory \cite[Ch.~5.4]{aubin2012differential}. Besides confirming this, Theorem~\ref{th:PWA} also provides a constructive way to find $\mc L_\alpha(\Phi,\mc S)$ as the maximal $\alpha$-norm of the linear gains associated to the $N^r$ polyhedral regions produced by \eqref{eq:online_control_2}, according to~\eqref{eq:Lipschitz_const_PWA}.

Thus, once established that also the state-to-input mapping in \eqref{eq:online_control_1} determines a continuous \gls{PWA} partition of $\mc{S}$, we next ask whether we can develop a technique compatible with the one proposed in \cite{fischetti2018deep,jordan2020exactly}. In fact, from those works we know that both the output and the Lipschitz constant of a \gls{ReLU} network, which we will use to approximate a traditional control law $\Phi(\cdot)$, can themselves be computed through an \gls{MILP}. We therefore require a constructive methodology compatible with that in \cite{fischetti2018deep,jordan2020exactly}, 
which will also allow us to compute the key quantity characterizing the approximation error $e(\cdot)$ derived in \S \ref{subsec:stability}. The next result answers this question:

\begin{theorem}\label{th:MILP}
	Let $\alpha \in \{1, \infty\}$ and $\Phi(\cdot)$ be defined as in \eqref{eq:online_control_1}.   The following statements hold true:
	\begin{itemize}
		\item[(i)] The state-to-input mapping $x \mapsto \Phi(x)$  can be computed via an \emph{\gls{MILP}};
		\item[(ii)] The Lipschitz constant $\mc{L}_{\alpha}(\Phi, \mc{S})$ can be computed via an \emph{\gls{MILP}}. \hfill$\square$
	\end{itemize}
\end{theorem}
\begin{proof}
	(i) From Theorem~\ref{th:PWA} we know that each facet of $\mc{S}$ can be associated with an \gls{mp-QP} in the form \eqref{eq:online_control_compact_2}. For a given $x \in \mc{S}^{(h)} \subseteq \mc{S}$, however, the \glspl{mp-QP} associated with facets other than the $h$-th one may not be feasible. Consider the following modified  collection of $p$ optimization problems:
	\begin{equation}\label{eq:online_control_3}
		\forall h \in \mc{P} : \left\{
		\begin{aligned}
			&\underset{z^{(h)}, \epsilon^{(h)}}{\textrm{min}} && \tfrac{1}{2} \|z^{(h)}\|^2_H + \tfrac{1}{2} \|\epsilon^{(h)}\|^2_Q + \varrho \|\epsilon^{(h)}\|_{\alpha}\\
			& \hspace{.2cm}\textrm{ s.t. } && C z^{(h)} \leq d +  T^{(h)} x + \epsilon^{(h)},\\
			&&& \epsilon^{(h)} \geq 0,
		\end{aligned}
		\right.
	\end{equation}
	where we have made a change of variables $z^{(h)} \defeq v^{(h)} + H^{-1}P^\top x$ and defined $T^{(h)} \defeq S^{(h)} + C H^{-1}P^\top$ to simplify notation.   In addition, we have also introduced the nonnegative slack vector $\epsilon^{(h)} \in \R^{\ell + pM}$ to ensure the feasibility of \eqref{eq:online_control_3} for all $x$ and for all $h \in \mc{P}$.
	
	We include both a quadratic and a polyhedral norm penalty on $\epsilon^{(h)}$ in the problem \eqref{eq:online_control_3}.   The 2-norm weighting matrix $Q \in \mathbb{S}^{\ell +pM}_{\succ 0}$ can be chosen arbitrarily and serves to ensure that the objective function is strictly convex in $\epsilon^{(h)}$.   According to \cite[Ch.~14.3]{fletcher2013practical}, instead, the weighting term $\varrho > 0$ can always be chosen to be sufficiently large such that solutions to \eqref{eq:online_control_3} coincide with solutions to \eqref{eq:online_control_compact_2} (modulo change of variables) for all cases in which the latter is feasible, i.e., the cost function turns into an \emph{exact penalty function}.
	
	Supposing without loss of generality that $\alpha = \infty$ (a similar argument holds for the case $\alpha = 1$), each optimization problem in \eqref{eq:online_control_3} can be
	turned into an equivalent \gls{mp-QP}:
	\begin{equation}\label{eq:online_control_4}
		\forall h \in \mc{P} : \left\{
		\begin{aligned}
			&\underset{z^{(h)}, \epsilon^{(h)}, t^{(h)}}{\textrm{min}} && \tfrac{1}{2} \|z^{(h)}\|^2_H + \tfrac{1}{2} \|\epsilon^{(h)}\|^2_Q + t^{{(h)}}\\
			& \hspace{.4cm}\textrm{ s.t. } && C z^{(h)} \leq d +  T^{(h)} x + \epsilon^{(h)},\\
			&&& 0 \leq \varrho \epsilon^{(h)} \leq t^{(h)} \bsone.
		\end{aligned}
		\right.
	\end{equation}
	For each $h \in \mc{P}$, we now write the associated \gls{KKT} conditions with nonnegative Lagrange multipliers $\mu^{(h)}$, $\nu^{(h)}$, $\xi^{(h)} \geq 0$ of appropriate dimensions:
	\begin{equation}\label{eq:KKT}
		\forall h \in \mc{P} : \left\{
		\begin{aligned}
			& z^{(h)} + H^{-1} C^\top \mu^{(h)}  = 0,\\
			& \epsilon^{(h)} - Q^{-1} (\mu^{(h)} - \varrho \nu^{(h)} + \varrho \xi^{(h)}) = 0,\\
			& 1 - \bsone^\top \nu^{(h)} = 0,\\
			& C z^{(h)} - d -  T^{(h)} x - \epsilon^{(h)} \leq 0,\\
			& \varrho \epsilon^{(h)} - t^{(h)} \bsone \leq 0, \ \epsilon^{(h)} \geq 0,\\
			& 0 \leq \mu^{(h)} \perp (C z^{(h)} - d -  T^{(h)} x - \epsilon^{(h)}),\\
			& 0 \leq \nu^{(h)} \perp (\varrho \epsilon^{(h)} - t^{(h)} \bsone), \ 0 \leq \xi^{(h)} \perp \epsilon^{(h)}.\\
		\end{aligned}
		\right.
	\end{equation}
	Using standard \gls{MI} modelling techniques \cite{heemels2001equivalence}, we can then remodel the nonlinear complementarity conditions in \eqref{eq:KKT} with equivalent \gls{MI} linear inequalities by means of binary variables $\sigma^{(h)}$, $\eta^{(h)}$, $\theta^{(h)}$ of suitable dimensions. These must satisfy (componentwise) the conditions: $[\sigma_i^{(h)} = 1] \implies [-C_{i,:} z^{(h)} + d_i +  T_{i,:}^{(h)} x + \epsilon_i^{(h)} = 0]$, $[\eta_i^{(h)} = 1] \implies [ - \varrho \epsilon_i^{(h)} + t^{(h)} = 0]$, and $[\theta_i^{(h)} = 1] \implies [\epsilon_i^{(h)} = 0]$.   Substituting these conditions in place of the complementarity conditions in \eqref{eq:KKT} leaves us with:
	\begin{equation}\label{eq:KKT_2}
		\forall h \in \mc{P} : \left\{
		\begin{aligned}
			& z^{(h)} + H^{-1} C^\top \mu^{(h)}  = 0, \ \bsone^\top \nu^{(h)} = 1,\\
			& \epsilon^{(h)} - Q^{-1} (\mu^{(h)} - \varrho \nu^{(h)} + \varrho \xi^{(h)}) = 0,\\
			& 0 \leq  -C z^{(h)} + d +  T^{(h)} x + \epsilon^{(h)} \leq \bar{\upsilon} (\bsone - \sigma^{(h)}),\\
			& 0 \leq -\varrho \epsilon^{(h)} + t^{(h)} \bsone\leq \bar{\omega}(\bsone - \eta^{(h)})\\
			& 0 \leq \epsilon^{(h)} \leq \bar{\epsilon}(\bsone - \theta^{(h)}), \ 0 \leq \mu^{(h)} \leq \bar{\mu} \sigma^{(h)},\\
			& 0 \leq \nu^{(h)} \leq \bar{\nu} \eta^{(h)}, \ 0 \leq \xi^{(h)} \leq \bar{\xi} \theta^{(h)}.\\
		\end{aligned}
		\right.
	\end{equation}
	This \gls{MI} reformulation is always feasible for any $x$ due to to the inclusion of the soft constraint terms $\epsilon^{(h)}$.  Since all of the sets in the original problem were compact, we are assured of finding positive constants $\bar{\mu}$, $\bar{\nu}$, $\bar{\xi}$, $\bar{\omega}$, $\bar{\epsilon}$ and $\bar{\upsilon}$ to upper-bound the inequalities in \eqref{eq:KKT_2}.
	
	Since the cost function in \eqref{eq:online_control_4} is not strictly convex in $t^{(h)}$, it is not immediately obvious whether its parametric solution is single valued anywhere.   However, we note that the objective \emph{is} strictly convex over that part of its domain in which the soft constraint terms $\epsilon^{(h)}$, and consequently $t^{(h)}$, are zero, and that this region will include the set $\mc S^{(h)}$.   The parametric solution will therefore be single valued over this region of interest.
	
	Given some $x \in \mc S$, for all $h \in \mc{P}$ let us now assume to have available a vector $z^{(h)}$ that satisfies the system in \eqref{eq:KKT_2} together with assorted other auxiliary variables (i.e., $\mu^{(h)}$, $\nu^{(h)}$, $\epsilon^{(h)}$, and so on).	We observe that the input-to-state mapping obtained from \eqref{eq:online_control_2} can hence be equivalently expressed as 
	\begin{equation}\label{eq:Phi_partial}
		\Phi(x) = \sum_{h \in \mc{P}} \delta^{(h)} v^{(h)} = \sum_{h \in \mc{P}} \delta^{(h)} (z^{(h)} - H^{-1}P^\top x),
	\end{equation}
	where $\delta^{(h)} \in \B$ is a binary variable that is meant to identify the region $\mc{S}^{(h)}$ in which the current $x$ resides.	In particular, we shall associate $\delta^{(h)} = 1$ to the $h$-th \gls{QP} instance such that i) it is strictly feasible (i.e., it does not need the slack vector, $\epsilon^{(h)}$), and ii) it identifies the facet of $\mc{S}$ that produces the largest feasible set in \eqref{eq:online_control_2}. These considerations can be summarized into the logical implication $[\delta^{(h)} = 1] \iff [\epsilon_i^{(h)} = 0, \text{ for all } i \in \{1,\ldots,\ell+pM\}] \wedge [F_{h,:} x \geq F_{i,:} x \text{ for all } i \in \mc{P}\setminus \{h\}]$ componentwise.
	
	The first predicate $[\epsilon_i^{(h)} = 0, \text{ for all } i \in \{1,\ldots,\ell+pM\}]$, which on the other hand is implied by $[\theta_i^{(h)} = 1, \text{ for all } i \in \{1,\ldots,\ell+pM\}]$, can be encoded by a binary variable $\kappa^{(h)} \in \B$ so that $[\kappa^{(h)}=1]\iff[\theta_i^{(h)} = 1, \text{ for all } i \in \{1,\ldots,\ell+pM\}]$, an implication that can be equivalently rewritten in terms of integer linear constraints as follows:
	\begin{equation}\label{eq:multiple_AND}
		\forall h \in \mc{P} :\left\{
		\begin{aligned}
			&-\theta_i^{(h)} + \kappa^{(h)} \le 0, \text{ for all } i \in \{1,\ldots,\ell+pM\},\\
			& \sum_{i=1}^{\ell+pM} \theta_i^{(h)}  - \kappa^{(h)} \le \ell+pM-1.
		\end{aligned}
		\right.
	\end{equation}
	The second predicate $[F_{h,:} x \geq F_{i,:} x \text{ for all } i \in \mc{P}\setminus \{h\}]$, instead, involves a comparison among terms linear in $x$, and can be equivalently encoded by a binary variable $\pi^{(h)}\in\B$ subject to a collection of additional \gls{MI} linear inequalities through standard integer modelling techniques.
	In particular, by letting $h=p$ (i.e., the last index in $\mc P$) without loss of generality, and by focusing on some $i \in \mc{P}\setminus \{p\}$, we define a scalar binary variable $\varsigma^{(p)}_i \in \B$ satisfying the implication $[(F_{i,:} - F_{p,:}) x \leq  0] \implies [\varsigma^{(p)}_i=1]$, which simply translates into \gls{MI} linear inequality $(F_{i,:} - F_{p,:}) x \ge \varepsilon -  (1+\varepsilon)\varsigma^{(p)}_i$, where $\varepsilon>0$ is a small tolerance (typically, the machine precision), beyond which the constraint is considered violated. Once introduced all the $p-1$ \gls{MI} inequalities, we have to impose further that $[\pi^{(p)}=1]\iff[\varsigma_i^{(p)}=1,\text{ for all } i \in \mc{P}\setminus \{p\}]$, which translates identically as in \eqref{eq:multiple_AND}:
	$$
	\forall h \in \mc{P} :\left\{
	\begin{aligned}
		&-\varsigma_i^{(p)} + \pi^{(p)} \le 0, \text{ for all } i \in \mc{P}\setminus \{p\},\\
		& \sum_{j=1}^{p-1} \varsigma_i^{(p)} - \pi^{(p)} \le p-2.
	\end{aligned}
	\right.
	$$
	Putting everything together, we will thus end up with an implication $[\delta^{(h)} = 1] \iff [\kappa^{(h)} = 1] \wedge [\pi^{(h)} = 1]$ in which the logical AND produces integer linear constraints as in \eqref{eq:multiple_AND}, i.e.,
	\begin{equation}\label{eq:multiple_AND_2}
		\forall h \in \mc{P} :\left\{
		\begin{aligned}
			&-\kappa^{(h)} +\delta^{(h)} \le 0, \\
			&-\pi^{(h)} + \delta^{(h)} \le 0, \\
			&\kappa^{(h)} + \pi^{(p)} - \delta^{(h)} \le 1.
		\end{aligned}
		\right.
	\end{equation}
	Finally, by introducing $p$ auxiliary variables $q^{(h)} \in \R^m$ such that $[\delta^{(h)} = 0] \implies [q^{(h)} = 0]$, $[\delta^{(h)} = 1] \implies [q^{(h)} = z^{(h)} - H^{-1}P^\top x]$, we obtain $\Phi(x) = \sum_{h \in \mc{P}} q^{(h)}$ provided that a collection of additional \gls{MI} constraints is satisfied and, specifically \cite{bemporad1999control}
	\begin{equation}\label{eq:MI_add}
		\forall h \in \mc{P} : \left\{
		\begin{aligned}
			&\underline{z} \delta^{(h)} \leq q^{(h)} \leq \bar{z} \delta^{(h)},\\
			&z^{(h)} - H^{-1}P^\top x - \bar{z}(1 - \delta^{(h)}) \leq q^{(h)} \leq z^{(h)} - H^{-1}P^\top x - \underline{z}(1 - \delta^{(h)}),
		\end{aligned}
		\right.
	\end{equation}
	for lower and upper bounds $\underline{z}$, $\bar{z}$ on $z^{(h)} - H^{-1}P^\top x$, assured to exist since $z^{(h)}$ is a linear function of the bounded variable $\mu^{(h)}$ and $x \in \mc S$, a C-polytope.
	The first part of the proof hence follows by noting that all the continuous and binary auxiliary variables, introduced to express the state-to-input mapping as $\Phi(x) = \sum_{h \in \mc{P}} q^{(h)}$, obey to \gls{MI} linear constraints.
	
	(ii) As we have observed in the first part of the proof,  solving a collection of \gls{MI} linear constraints involving a number of continuous and binary auxiliary variables yields $\Phi(x) = \sum_{h \in \mc{P}} \delta^{(h)} (z^{(h)} - H^{-1}P^\top x)=\sum_{h \in \mc{P}} q^{(h)}$.
	In view of Theorem~\ref{th:PWA}, note that this additionally corresponds to an affine control policy $\Phi(x) = K(\kappa^{(h)},\sigma^{(h)})x+c(\kappa^{(h)},\sigma^{(h)})$, for some gain matrix $K(\kappa^{(h)},\sigma^{(h)})\in\R^{m\times n}$ and vector $c(\kappa^{(h)},\sigma^{(h)})\in\R^m$, which are uniquely (in view of Standing Assumption~\ref{standing:LICQ}) determined by the values $\kappa^{(h)}$ and $\sigma^{(h)}$ assume. In fact, those terms together identify the case in which the $h$-th instance is feasible (i.e., $\epsilon_i^{(h)} = 0$ for all $i$) and which one of the constraints in \eqref{eq:online_control_4} is active (i.e., $-C_{i,:} z^{(h)} + d_i +  T_{i,:}^{(h)} x = 0$). Note that we have needed to characterize explicitly neither $K(\cdot,\cdot)$ nor $c(\cdot,\cdot)$. To simplify exposition, in the rest of the proof we will employ $K$ and $c$, thus omitting the double argument.
	
	We now resort the \gls{KKT} systems in \eqref{eq:KKT_2}, which collect necessary and sufficient conditions characterizing both the feasibility and optimality of each \gls{QP} in \eqref{eq:online_control_4}. In fact, to compute $\mc{L}_{\alpha}(\Phi,\mc S)$, $\alpha \in \{1, \infty\}$ without explicit calculation of $\Phi(\cdot)$ across all of its partition regions, we first perturb $x$ along the canonical basis  vectors in $\R^n$ and consider how the optimal solutions to the collection of \glspl{QP} in \eqref{eq:online_control_4} vary, provided that the same set of active/inactive constraints is imposed, according to the binary vector $\sigma^{(h)}$. For all $h\in\mc P$, we then introduce real auxiliary variables $\{x^i, z^{(h),i}, \epsilon^{(h),i}, t^{(h),i}, \mu^{(h),i}, \nu^{(h),i}, \xi^{(h),i}, \eta^{(h),i}, \theta^{(h),i}\}$ for $i = 1, \ldots, n$, and additional \gls{MI} linear constraints as
	\begin{equation}\label{eq:KKT_perturbed}
		\forall h\in\mc P, i \in \{1, \ldots, n\} : \left\{
		\begin{aligned}
			& x^i = x + e^i,\\
			& z^{(h),i} + H^{-1} C^\top \mu^{(h),i}  = 0, \ \bsone^\top \nu^{(h),i} = 1,\\
			& \epsilon^{(h),i} - Q^{-1} (\mu^{(h),i} - \varrho \nu^{(h),i} + \varrho \xi^{(h),i}) = 0,\\
			& -\bar{\upsilon} (\bsone - \sigma^{(h)}) \leq  -C z^{(h),i} + d +  T^{(h)} x^i + \epsilon^{(h),i} \leq \bar{\upsilon} (\bsone - \sigma^{(h)}),\\
			& 0 \leq -\varrho \epsilon^{(h),i} + t^{(h),i} \bsone\leq \bar{\omega}(\bsone - \eta^{(h),i})\\
			& 0 \leq \epsilon^{(h),i} \leq \bar{\epsilon}(\bsone - \theta^{(h),i}), \ -\bar{\mu} \sigma^{(h)} \leq \mu^{(h),i} \leq \bar{\mu} \sigma^{(h)},\\
			& 0 \leq \nu^{(h),i} \leq \bar{\nu} \eta^{(h),i}, \ 0 \leq \xi^{(h),i} \leq \bar{\xi} \theta^{(h),i}.
		\end{aligned}
		\right.
	\end{equation}
	where $e^i \in \R^n$ is the $i$-th vector of the canonical basis. Note that the nonnegativity of both $-C z^{(h),i} + d +  T^{(h)} x^i + \epsilon^{(h),i}$ and $\mu^{(h),i}$ is relaxed to guarantee the existence of a solution to \eqref{eq:KKT_perturbed}, since $x^i$ may fall within a critical region i) originating from a possibly different \gls{QP} instance in \eqref{eq:online_control_4}, and ii) with active constraints that differ from those for $x$, identified by $\sigma^{(h)}$. In addition, note that only the state $x$, which serves as a parameter, is varied to obtain $x^i$, while the newly introduced $\{x^i, z^{(h),i}, \epsilon^{(h),i}, t^{(h),i}, \mu^{(h),i}, \nu^{(h),i}, \xi^{(h),i},\eta^{(h),i}, \theta^{(h),i}\}$ are additional decision variables, subject to the \gls{MI} linear constraints in \eqref{eq:KKT_perturbed}, which allow us to define
	$$
	\forall i \in \{1, \ldots, n\} : \Phi^i = \sum_{h \in \mc{P}} \delta^{(h),i}(z^{(h),i} - H^{-1}P^\top x^i).
	$$
	Each binary variable $\delta^{(h),i} \in \B$, for all $h\in \mc P$ and $i \in \{1, \ldots, n\}$, is introduced to parallel the effect of $\delta^{(h)}$ in the first part of the proof. Specifically, $\delta^{(h),i}$ is required to satisfy the implication $[\delta^{(h),i} = 1] \iff [\epsilon_j^{(h),i} = 0, \text{ for all } j \in \{1,\ldots,\ell+pM\}] \wedge [F_{h,:} x^i \geq F_{j,:} x^i \text{ for all } j \in \mc{P}\setminus \{h\}]$, and therefore allows one to identify only the control contribution $z^{(h),i} - H^{-1}P^\top x^i$ associated with that strictly feasible instance among those in \eqref{eq:KKT_perturbed} producing the largest feasible set. 
	Adopting an identical reformulation as in part i) of the proof thus leaves us with:
	\begin{equation}\label{eq:single_control}
		\forall i \in \{1, \ldots, n\} : \Phi^i = \sum_{h \in \mc{P}} q^{(h),i}.
	\end{equation}
	Now, observe that $\Phi^i$ may differ from $\Phi(x^i) = \Phi(x + e^i)$, since the active set encoded by $\sigma^{(h)}$ for some $x$ may differ from the active set at the perturbed point $x + e^i$, particular for those $x$ near the boundary of their partition.   However, it still holds that $\Phi^i = Kx^i+c$, and hence that
	\begin{align*}
		\left[ \Phi^1 \, \cdots \, \Phi^n \right]
		=K \left[ x^1 \, \cdots \, x^n \right] + c \otimes \bsone^\top
		&=K \left[x \otimes \bsone^\top + \left[ e^1 \, \cdots \, e^n \right]\right] + c \otimes \bsone^\top\\
		&= \left(K x + c\right) \otimes \bsone^\top + K
	\end{align*}	
	and we can isolate the gain term $K$ directly to obtain
	\begin{equation}\label{eq:gain_expression}
		K =  \left[ \Phi^1 \, \cdots \, \Phi^n \right]- \Phi(x)\otimes \bsone^\top=\left[ \sum_{h \in \mc{P}} (q^{(h),1}-q^{(h)}) \, \cdots \, \sum_{h \in \mc{P}} (q^{(h),n}-q^{(h)}) \right].
	\end{equation}	
	We have therefore constructed an expression for the controller gain in the critical region parametrized by some choice of $\sigma^{(h)}$ associated with that strictly feasible \gls{QP} in \eqref{eq:online_control_4} producing the largest feasible set, which can itself be computed numerically for any $x \in \mc{X}$. After combining all of the additional variables and constraints introduced in this part of the proof, we can finally apply the technical result in \cite[Prop.~5.2]{fabiani2021reliably} to construct \pgls{MILP} suitable to compute $\mc{L}_{\alpha}(\Phi, \mc{S})$, for some given $\alpha \in \{1, \infty\}$.
\end{proof}

\section{Certifying the approximation quality of ReLU-based controllers}\label{sec:certificates}
Once developed analytical tools to compute key quantities characterizing $\Phi(\cdot)$ as defined in \eqref{eq:vertex_law_1}--\eqref{eq:vertex_law_2}, or \eqref{eq:online_control_1}, we are then able to calculate exactly the worst-case approximation error to apply the reliability certificate in Proposition~\ref{prop:NN_stability}.

Since the existence of a polyhedral \gls{CLF} for a polytopic system is a necessary and sufficient condition for its exponential stabilizability inside $\mc{S}$ \cite[Prop.~7.39]{blanchini2015set}, and since $\Phi(\cdot)$ as defined by any of the methods described in \S\ref{sec:problem_description} makes $\Psi(\cdot)$ a Lyapunov function for the closed-loop dynamics in \eqref{eq:polytopic} with $u(x) = \Phi(x)$, we may expect that the \gls{ReLU} based policy should also be stabilizing when the approximation error $e(\cdot) = \Phi_\normaltext{\textrm{NN}}(\cdot) - \Phi(\cdot)$ is small enough.  Given the results in \S \ref{sec:geometric_char}, this error mapping stems from the difference of two continuous \gls{PWA} mappings, and so is also \gls{PWA} continuous \cite[Prop.~1.1]{gorokhovik1994pwa}. 
The error can therefore also be shown to be bounded and Lipschitz continuous on $\mc{S}$, and we can apply the result of \S \ref{subsec:stability} to find a condition under which uniform ultimate boundedness is guaranteed.

The following result provides an offline, \gls{MI} optimization-based procedure to compute exactly both the worst-case approximation error between the two \gls{PWA} mappings, $\bar{e}_{\alpha} \defeq \textrm{max}_{x \in \mc{S}} \ \|e(x)\|_\alpha$, and the associated Lipschitz constant over $\mc{S}$, $\mc{L}_\alpha(e, \mc{S})$, for $\alpha \in \{1, \infty\}$.
\begin{theorem}\label{th:error_comp}
	For $\alpha \in \{1, \infty\}$, the approximation error $e(\cdot)$ has the following properties:
	\begin{itemize}
		\item[i)] The maximal error $\bar{e}_{\alpha}$ can be computed by solving an \normaltext{\gls{MILP}};
		\item[ii)] The Lipschitz constant $\mc{L}_\alpha(e, \mc{S})$ can be computed by solving an \normaltext{\gls{MILP}}. \hfill$\square$
	\end{itemize}
\end{theorem}
\begin{proof}
	i) This part follows by observing that the control action produced by the stabilizing policy $\Phi(\cdot)$ is linear in the state $x$ (if defined as in \eqref{eq:vertex_law_1}), or it can be manipulated to obtain a linear function of variables satisfying a collection of state-dependent \gls{MI} linear inequalities (if defined as in \eqref{eq:vertex_law_2}, or \eqref{eq:online_control_1} from Theorem~\ref{th:PWA}.(i)). This is consistent with the results in \cite{fischetti2018deep,jordan2020exactly} characterizing the output produced by a \gls{ReLU} network in general position as in \eqref{eq:RELU_NN}, which can be encoded likewise. Thus, the proof concludes by relying on \cite[Prop.~5.2]{fabiani2021reliably} after noting that, in all the considered cases, a vector norm maximization problem is a special case of the matrix maximization one.
	
	\sloppy ii) Since $e(\cdot)$ is a continuous \gls{PWA} mapping \cite[Prop.~1.1]{gorokhovik1994pwa}, we have $\mc{L}_\alpha(e, \mc{S}) = \textrm{max}_{x \in \mc{S}} \ \|K_{e}(x)\|_\alpha = \textrm{max}_{x \in \mc{S}} \ \|K_\textrm{NN}(x) - K(x)\|_\alpha$,
	where $K_\textrm{NN}(\cdot)$ denotes the local linear gain of the \gls{ReLU} network in \eqref{eq:RELU_NN}, and $K(\cdot)$ that of the \gls{PWA} controller in \eqref{eq:vertex_law_1}, \eqref{eq:vertex_law_2} or \eqref{eq:online_control_1}. While the linear gains for \eqref{eq:vertex_law_1} are known once the simplicial partition of $\mc{S}$ has been fixed, the ones characterizing \eqref{eq:vertex_law_2} and \eqref{eq:online_control_1} (from Theorem~\ref{th:MILP}.(ii)) can be obtained as a linear expression of variables subject to \gls{MI} linear constraints. This is again compatible with the way of expressing the Jacobian of $F(\cdot)$ over $\mc{S}$ \cite[Appendix~D]{jordan2020exactly}, and hence $\textrm{max}_{x \in \mc{S}} \ \|K_\textrm{NN}(x) - K(x)\|_\alpha$ can be suitably recast into the framework of \cite[Prop.~5.2]{fabiani2021reliably}, showing that the norm of a matrix whose entries are affine in the state variable can be computed through \pgls{MILP}.
\end{proof}
The first quantity is precisely of the type required to apply the stability result of \S \ref{subsec:stability}, thus supplying a condition on the optimal value of an \gls{MILP} sufficient to certify the uniform ultimate boundedness of the closed-loop system \eqref{eq:polytopic} under the action of $\Phi_\textrm{NN}(\cdot)$, obtained by suitably training a \gls{ReLU} network to replicate $\Phi(\cdot)$. 
Whether or not $\bar{e}_{\alpha}$ can be made sufficiently small to meet that certificate depends on the amount and quality of training data, as well as on the complexity of the network (i.e., number of neurons and layers). These considerations will be further explored in the next section.

\vspace{-.8cm}
{
	\section{Constraints satisfaction, network complexity and discussion}
}

We now discuss several practical aspects related with the design of a \gls{ReLU} network to compute a controller approximation $\Phi_{\textrm{NN}}(\cdot)${, also including further numerical simulations and discussing potential practical limitations of the proposed approach}.

\bigskip

\subsection{On the input constraints satisfaction}\label{subsec:input_cons}
We now discuss several options to force our \gls{PWA-NN} controller to meet input constraints, i.e., $\Phi_{\normaltext{\textrm{NN}}}(x) \in \mc U$ for all $x \in \mc S$. Note that none of the techniques we are about to introduce affect the results developed in the paper, as they hold true for any trained \gls{ReLU} network regardless the training methodology employed and the post-training performed.

\subsubsection{Structural fulfilment} 
To start, we note that designing a control proxy $\Phi_{\textrm{NN}}(\cdot)$ satisfying input constraints can even be enforced during the derivation of a traditional controller with any of the strategies presented in \S \ref{sec:problem_description}. In fact, by focusing on the augmented state variable $\hat{x} \defeq \col(x, u)$, we obtain the following discrete-time uncertain dynamics
\begin{equation}\label{eq:augmented_dyn}
	\begin{aligned}
		\hat{x}^+ &=
		\begin{bmatrix}
			A(w) & 0\\
			0 & 0
		\end{bmatrix} \hat{x} + \begin{bmatrix}
			B(w)\\
			I
		\end{bmatrix} u\\
		&= \Bigg( \sum_{i \in \mc M} w_i
		\underbracket{\begin{bmatrix}
				A_i & 0\\
				0 & 0
		\end{bmatrix}}_{\rdefeq \hat A_i} \Bigg) \hat{x} + \Bigg( \sum_{i \in \mc M} w_i \underbracket{\begin{bmatrix}
				B_i\\
				I
		\end{bmatrix}}_{\rdefeq \hat B_i}\Bigg) u,
	\end{aligned}
\end{equation}
which still enjoys a polytopic structure in view of the properties of the set $\mc W \ni w$. This hence enables us to incorporate input constraints as state ones directly, since $\hat{x} \in \mc X \times \mc U$. 
Thus, mimicking the mathematical developments in \S \ref{subsec:stability} with \eqref{eq:augmented_dyn} in place of \eqref{eq:polytopic}  allows one to derive a (possibly more conservative) bound depending on $\zeta$ to conclude on the robust positively invariance of some joint set $\hat{\mc S} \subseteq \mc S \times \mc{U}$ for the underlying perturbed dynamics, and hence for \eqref{eq:augmented_dyn} with some control surrogate $\Phi_{\textrm{NN}}(\cdot)$. 

\subsubsection{Training and output verification}
Another possibility to enforce input constraints requires one to design a tailored training phase for the \gls{NN} at hand. In particular, the generation of adversarial examples and traditional output verification techniques have been widely employed to produce training signals able to guide the network's output to meet specific requirements \cite{mirman2018differentiable,wong2018provable}. Along the same line, procedures to determine the weights of a \gls{NN}, tailored for the satisfaction of polyhedral input constraints \cite{markolf2021polytopic}, and based on the systematic manipulation of the gradient of $\Phi_\textrm{NN}(\cdot)$ when samples approaches the boundary of the feasible set \cite{ICLR16_hausknecht}, were given. Once the \gls{NN} is trained, instead, one may also resort on reachable set-based approaches \cite{NEURIPS2018_be53d253,9301422,karg2020stability}  to verify whether or not a \gls{NN} output falls within a desired set. For instance, in \cite{9301422} was shown that, due to the specific properties of \glspl{ReLU}, such a verification step simply amounts to solve a convex program.

\subsubsection{Surrogate controller deployment}
We finally suggest an a-posteriori way to enforce input constraints, i.e., during the controller deployment, which is achieved through a systematic projection onto the set $\mc U$ of feasible inputs (a saturation, in case of box constraints). 
In fact, by noting that the stability analysis involving the perturbed system \eqref{eq:perturbed_dyn} holds for any state-uncertainty-dependent disturbance $d(x,w) = B(w)(\Phi(x) - \Phi_{\textrm{NN}}(x))$, bounded in some norm $\alpha \in \{1, \infty\}$ by $\bar{e}_\alpha$ due to Proposition~\ref{prop:NN_stability}.(i), and in view of the nonexpansiveness of the projection mapping \cite[Cor.~12.20]{rockafellar2009variational}, we have
$$
\begin{aligned}
	&\textrm{max}_{x \in \mc{X}}  \, \|\Phi_{\textrm{NN}}(x) - \Phi(x)\|_\alpha \le \bar{e}_\alpha,\\
	& \implies \|\textrm{proj}_{\mc U}(\Phi_{\textrm{NN}}(x) - \Phi(x))\|_\alpha \le \bar{e}_\alpha, \text{ for all } x \in \mc{X},\\
	& \implies \|\textrm{proj}_{\mc U}(\Phi_{\textrm{NN}}(x)) - \Phi(x)\|_\alpha \le \bar{e}_\alpha, \text{ for all } x \in \mc{X},
\end{aligned}
$$
where we also exploited the relation $\Phi(x) = \textrm{proj}_{\mc U} (\Phi(x))$. This implies that the theory developed in \S \ref{subsec:stability} actually supports the safe implementation of $u(x) = \textrm{proj}_{\mc U} (\Phi_{\textrm{NN}}(x))$ while guaranteeing the closed-loop stabilization of the considered polytopic system in \eqref{eq:polytopic}.

\subsection{Bounding the complexity of \gls{PWA} neural-network controllers}\label{subsec:complexity}
The \gls{NN} complexity, characterized by its \emph{depth} $L$ (the number of hidden layers) and \emph{width} $\bar{N}$ (the number of neurons, for simplicity identical through the layers) plays a fundamental role to reproduce a given $\Phi(\cdot)$ as faithfully as possible. Given the intrinsic nature of a \gls{ReLU} network, however, whose output is known to be susceptible to overparametrization \cite{neyshabur2018the}, we aim at exploiting the geometry of the proposed approximation to bound the required complexity of $\Phi_{\textrm{NN}}(\cdot)$ in \eqref{eq:RELU_NN}, with the goal of keeping it reasonably low.

It is well-known that a \gls{ReLU} network with one hidden layer can approximate arbitrarily well any continuous function over a compact subset of $\R^n$ with a sufficient number of neurons \cite{cybenko1989approximation,hornik1989multilayer}. Specifically, this type of structure produces an input-output mapping in the so-called \emph{canonical piecewise-linear representation} \cite[Def.~2]{chua1988canonical}. However, unless $\Phi(\cdot)$  possesses certaint structural properties a \gls{ReLU} network of finite complexity is insufficient for \emph{exact} reproduction purposes \cite{petersen2018optimal}, even in the case of the simple vertex control law in \eqref{eq:vertex_law_1} resulting from a simplicial partition of $\mc{S}$ as in Fig.~\ref{fig:vert_simplices}. Let us consider indeed the following key concepts:

\begin{definition}\textup{(Order of intersections \cite{kahlert1992complete})}\label{def:intersections_order}
	Given a partition of $\mc{S} \subset \R^{n}$, an $(n-1)$-dimensional hyperplane (boundary) is said to be a first-order intersection, denoted by $S^{(1)}$. A linear manifold of dimension $n-k$ in $\mc{S}$ is a \emph{$k$-th order intersection $S^{(k)}$} if it is the intersection of at least two linear manifolds of type $S^{(k-1)}$, i.e., $S^{(k)} \defeq \cap_{i=1}^{\geq 2} S_{i}^{(k-1)}$.
	\hfill$\square$
\end{definition}
\begin{definition}\textup{(Consistent variation property \cite{chua1988canonical})}\label{def:consistent_variation}
	\sloppy A \normaltext{\gls{PWA}} mapping $F: \mc{F} \to \R^m$ possesses the \emph{consistent variation property} if for each boundary $q$, parametrized as $\set{x \in \R^n}{\alpha_{q}^\top x + \beta_{q} = 0}$, there exist constant vectors $\gamma_{q}, \delta_{q} \in \R^{m}$ such that, for any pair of components $F_{p}(x) = G_{p} x+ g_{p}$ and $F_{p'}(x) = G_{p'} x+ g_{p'}$ defined in two regions separated by such a boundary, it holds that $G_{p}-G_{p'}=\gamma_{q} \alpha_{q}^\top$ and $g_{p}-g_{p'}=\delta_{q}$.
	\hfill$\square$
\end{definition}
The consistent variation property is a necessary and sufficient condition ensuring that a canonical piecewise-linear representation of a \gls{PWA} mapping \cite[Th.~1]{chua1988canonical} will exist. Our controllers $\Phi(\cdot)$ are therefore only amenable to exact representation by means of a \gls{ReLU} network with one hidden layer when they possess this property, which is generally a difficult requirement to ensure. In $\R^2$, for instance, the consistent variation property is violated if the highest order of intersection $k = n = 2$ is generated by more than two hyperplanes intersecting in one point \cite[Ex.~3]{chua1988canonical}, thereby producing a \emph{degenerate} partition. Therefore, none of the cases in Fig.~\ref{fig:vert_simplices} or \ref{fig:th_PWA} satisfy such a condition. 
The result given next allow us to exploit the geometry of our approximation problem to provide a bound on the complexity of the \gls{ReLU} network in \eqref{eq:RELU_NN}:

\begin{proposition}\label{prop:complexity_bound}
	Let $N^r$ be the total number of regions produced by $\Phi(\cdot)$. To reproduce $\Phi(\cdot)$ exactly, {it is necessary that the \normaltext{\gls{ReLU}} network in \eqref{eq:RELU_NN} is composed of at most $L = \lceil \log_2 (n+1) \rceil$ hidden layers with $\bar{N} \ge n$ neurons each, where $\bar{N}$ satisfies:
\begin{equation}\label{eq:neurons_condition}
	 \left[\sum_{i=0}^{n} \binom{\bar N}{i}\right]^{\lceil \log_2 (n+1) \rceil} \geq N^r.
\end{equation}
}
\hfill$\square$
\end{proposition}
\begin{proof}
{
	From \cite[Th.~2.1]{arora18understanding}, we know that every continuous \gls{PWA} mapping can be represented by a \gls{ReLU} network with at most $\lceil \log_2 (n+1) \rceil$ layers.
	 Moreover, from \cite[Th.~3.13]{montufar2022sharp}, a \gls{ReLU} network with $n$ inputs, $m$ outputs and $L=\lceil \log_2 (n+1) \rceil$ hidden layers of width $n_j \geq n$, for all $j \in \left\{1, \ldots, L\right\}$, produces the following number of affine regions:
	 $$
		 \prod_{j=1}^{\lceil \log_2 (n+1) \rceil} \sum_{i=0}^{e_j} \binom{n_j}{i},
 	$$
 	where $e_j \defeq \textrm{min} \{n, n_1, \ldots, n_{j-1}\} = n$, $j \in \{1, \ldots, \lceil \log_2 (n+1) \rceil\}$.
	Thus, the relation in \eqref{eq:neurons_condition} directly follows by comparison with the total number of regions characterizing $\Phi(\cdot)$, $N^r$, and by considering hidden layers of fixed width $n_j = \bar{N} \ge n$ so that $e_j=n$, for all $j \in \{1, \ldots, \lceil \log_2 (n+1) \rceil\}$, for the \gls{ReLU} network in \eqref{eq:RELU_NN}.
}
\end{proof}
\begin{table}[!t]
	\caption{Complexity bounds for the examples from the literature}\label{tab:examples_complexity}
	\centering
	\begin{tabular}{cccccc}
		\toprule  
		Ex. & $N^r$ & $L ={\lceil \log_2 (n+1) \rceil}$ & $\bar{N}$ & $\{n_j\}_{j = 1}^{L}$\\
		\midrule
		(a)   & $14$ & \multirow[c]{2}{*}{$2$} & ${2}$ & ${\{2,2\}}$\\
		(f) & $21$ & & ${3}$ & ${\{3,2\}}$\\
		\midrule
		(b)  & $12$ & \multirow[c]{2}{*}{$2$} & ${2}$ & ${\{2,2\}}$\\
		(g)  & $26$ & & ${3}$ & ${\{3,2\}}$\\
		\midrule
		(c) & $48$ & \multirow[c]{2}{*}{$2$} & ${3}$ & ${\{3,3\}}$\\
		(h)  & $61$ & & ${4}$ & ${\{4,3\}}$\\
		\midrule
		(d) &  $24$ & \multirow[c]{2}{*}{$2$} & ${3}$ & ${\{3,2\}}$\\
		(i)  & $40$ & & ${3}$& ${\{4,2\}}$\\
		\midrule
		(e)  & $38$ & \multirow[c]{2}{*}{$2$} & ${3}$ & ${\{4,2\}}$\\
		(j) & $66$ & & ${4}$ & ${\{4,3\}}$\\
		\bottomrule
	\end{tabular}
\end{table}

\begin{continuance}{ex:numerics}
	For all the examples in Tab.~\ref{tab:examples}, we have $n = 2$ and $m = 1$, while $\bar{k} = 2$, which coincides with the dimension of the state-space. Thus, to reproduce any of our controllers $\Phi(\cdot)$ exactly, the \normaltext{\gls{ReLU}} network in \eqref{eq:RELU_NN} should have $L = {\lceil \log_2 3 \rceil = 2}$ hidden layers with at least $\bar{N} \geq 2$ neurons each corresponding to the numerical values reported in the second-last column of Tab.~\ref{tab:examples_complexity}, obtained by solving \eqref{eq:neurons_condition}. Finally, the last column shows that relaxing the fixed width condition of the hidden layers may result in a tighter difference between the number of regions produced and the required ones (e.g., cases {\normaltext{(g)}} and {\normaltext{(f)}}), or a  smaller number of total neurons (e.g., cases {\normaltext{(d)}} and {\normaltext{(j)}}).
	\hfill$\square$
\end{continuance}


{
\subsection{Further numerical results and discussion}\label{subsec:limitations}
We use the examples in Tab.~\ref{tab:examples} to verify Theorem~\ref{th:error_comp} numerically, also discussing potential limitations of the proposed methodology to design fast \gls{ReLU}-based proxies for ultimate boundedness control of polytopic systems.  Simulations are run in Matlab using Gurobi \cite{gurobi} as an \gls{MILP} solver on a laptop with a Quad-Core Intel i5 2.4 GHz CPU and 8 Gb RAM. The \gls{ReLU} networks are all trained by adopting a Levenberg–Marquardt algorithm with mean squared normalized error as a performance function.

\begin{table*}[!t]
	\caption{Numerical results for the examples in Tab.~\ref{tab:examples}}\label{tab:num_results}
	\centering
	\begin{tabular}{ccccccc}
		\toprule
		Ex. & $|\mc P|$ & $\bar{e}_\infty$ & $\zeta$ & CPU time [s] & $b$\\
		\midrule
		(a) &\multirow[c]{2}{*}{$6$}&  $7.09\times10^{-4}$ &  \multirow[c]{2}{*}{$0.7085$} & $77.6370$ & $0.0098$\\
		(f) &  & $9.15\times10^{-4}$ & & $75.9037$ & $0.0126$\\
		\midrule
		(b) &\multirow[c]{2}{*}{$8$}& $9.75\times10^{-4}$ &  \multirow[c]{2}{*}{$0.0799$} & $79.3175$ & $0.0122$\\
		(g) &  & $9.77\times10^{-4}$ && $79.0021$ & $0.0134$\\
		\midrule
		(c) &\multirow[c]{2}{*}{$30$}& $0.0064$ &  \multirow[c]{2}{*}{$0.0832$} & $161.0301$ & $0.0770$\\
		(h) &  & $0.0071$ && $163.5892$ & $0.0854$\\
		\midrule
		(d) &\multirow[c]{2}{*}{$8$}& $0.0240$ &  \multirow[c]{2}{*}{$0.9448$} & $81.6634$ & $0.0254$\\
		(i) &  & $0.0255$ && $80.0135$ & $0.0270$\\
		\midrule
		(e) &\multirow[c]{2}{*}{$18$}& $0.0059$ &  \multirow[c]{2}{*}{$0.0726$} & $112.8425$ & $0.0812$\\
		(j) &  & $0.0063$ && $112.3749$ & $0.0867$\\
		\bottomrule
	\end{tabular}
\end{table*}

To set the complexity of each \gls{ReLU} network, we have followed the indications reported in the corresponding rows of Tab.~\ref{tab:examples_complexity}, where we have considered the case of equally distributed neurons across the $L$ hidden layers. Note that, in accordance to Tab.~\ref{tab:examples_complexity} indeed, we are implicitly designing minimum complexity \gls{ReLU} networks with the same number of layers $L=2$, each of them composed of almost comparable number of neurons (from $2$ to $4$). As a consequence, the training time required to determine the values of the underlying \gls{NN} weights is almost equivalent for all the considered configurations (i.e., from $8.37$ to $9.89$[s]).

For all the considered examples, the upper bound $\zeta$ is computed from \eqref{eq:WC_bound} with $\rho = 0.9999$. The numerical results obtained are summarized in Tab.~\ref{tab:num_results} where we have considered the case $\alpha = \infty$, trained each \gls{ReLU} network with $50 \times 10^3$ samples obtained from the \gls{PWA}, minimal selection-based control law expressed in \eqref{eq:online_control_2} with $P=0$ and $H=I$. 

By analyzing the results in Tab.~\ref{tab:num_results} -- specifically, by contrasting the third and fourth column -- we notice that we have always succeeded in the design of a minimum complexity, stabilizing \gls{ReLU}-based surrogate $\Phi_{\textrm{NN}}(\cdot)$ of $\Phi(\cdot)$ in \eqref{eq:online_control_2} for all the considered cases, i.e., Ex.~(a)--(j). In particular, the resulting values for $\bar{e}_\infty$ suggest that the neighbourhood of the origin we are assured to reach with $\Phi_{\textrm{NN}}(\cdot)$ can be made very small in practise, certifying to ultimately bounding the system state in a set up to $99.02\%$ smaller than the original volume of the control invariant set $\mc S$ (values for $b$, last column). Note that the obtained results can, in principle, be further improved by adding extra layers or neurons in the architecture underlying $\Phi_{\textrm{NN}}(\cdot)$ -- this may come at the price of slightly increasing both the training time and the time required for computing $\bar{e}_\infty$.

On the other hand, a well-known drawback of \gls{MI} optimization is poor scalability with increasing problem size. While it has already been observed in \cite{fabiani2021reliably} that the computation time of the worst-case error $\bar{e}_\infty$ is only weakly dependent on the state dimension $n$, whereas particularly sensitive to the number of inputs $m$, here we want to analyze the impact the number of facet of $\mc{S}$ has on our methodology. This is particularly relevant since, compared to the technique presented in \cite{fabiani2021reliably}, i) one needs to consider several \gls{KKT} systems as in \eqref{eq:KKT}, precisely one for each facet of $\mc S$, i.e., $h \in \mc P$, and ii) the expression for $\Phi(\cdot)$ in \eqref{eq:Phi_partial} depends on the summation of a bunch of continuous/binary variables over the facet $h \in \mc P$. As a direct consequence of both of these facts, one hence has to include in the derivation of the resulting \gls{MILP} a non-negligible number of additional \gls{MI} auxiliary variables (Theorem~\ref{th:MILP}). Accordingly, the complexity of the polyhedral representation of $\mc S$ plays a key role in the computational analysis of our approach. With this regard, Tab.~\ref{tab:num_results} indicates that, as expected, the computation time is indeed quite sensitive to the number of facet characterizing the representation of the polyhedral control invariant set $\mc S$, as is evident by contrasting the fifth column for scenarios (a)/(b)/(d)--(e), and (e)--(c). In addition, we recall that bottom rows of each pair of examples are associated with a richer \gls{PWA} partition -- see Fig.~\ref{fig:th_PWA}. Thus, it is clear that the performance, in terms of approximation quality, are comparable and do not deteriorate too much for each pair of examples, along with the computation time, which is almost equivalent for both the considered pair of cases ((a)--(f), (b)--(g), and so on).

In conclusion, the non-negligible offline computational efforts exhibited by our certification  method, which do not scale particularly well with the complexity of the polyhedral representation of the control invariant set $\mc S$, as expected, enables the design of stabilizing \gls{ReLU}-based controllers with minimum complexity, which can be implemented and evaluated on dedicated hardware up to tens of MHz \cite{zhang2019real,schindler2020real}.
}

\section{Conclusion}
We have considered the design of \gls{ReLU}-based approximations of traditional controllers for polytopic systems, enabling implementation even on very fast embedded control systems. We have shown that our reliability certificate require one to construct and solve an \gls{MILP} offline, whose associated optimal value characterizes a quantity of the approximation error. While a variable structure controller enjoys a nice geometric structure by construction, this is not immediate for a (minimal) selection-based policy. We have hence shown that the latter also introduces a state-to-input \gls{PWA} mapping, and provided a systematic way to encode both its output and the maximal gain through binary and continuous variables subject to \gls{MI} constraints. This optimization-based result is compatible with existing results from the machine learning literature on computing the output of a trained \gls{ReLU} network. Taken together they provide a sufficient condition to assess the reliability, in terms of uniform ultimate boundedness of the closed-loop system, of a given \gls{ReLU}-based approximation traditional controllers for uncertain systems.
{In the spirit of \cite{fabiani2021reliably}, future work will focus on the derivation of suitable certificates, possibly involving the Lipschitz constant of the approximation error, $\mc{L}_\alpha(e, \mc S)$, to recover full exponential stabilization of polytopic systems with neural network-based control surrogates.}

\bibliographystyle{plain}
\bibliography{23_IJRNC}


%

\end{document}